\documentclass{lmcs} %%% last changed 2014-08-20
\pdfoutput=1

\usepackage{lastpage}
\renewcommand{\dOi}{14(4:19)2018}
\lmcsheading{}{1--\pageref{LastPage}}{}{}%
{Mar.~06,~2018}{Nov.~26,~2018}{}


\usepackage{hyperref}

\usepackage{tikz}
\usepackage{stmaryrd}
\usepackage{amsmath}
\usepackage{amssymb}
\usepackage{subcaption}
\usepackage{ebproof}
\usepackage{forest}
\usepackage{calc}
\usepackage{graphicx}
\usepackage[boxed,linesnumbered]{algorithm2e}

\newcommand{\transm}[2]{\theta_{\mathcal{M}}\left(#1\right)\left(#2\right)}
\newcommand{\trans}[2]{\theta \left(#1\right)\left(#2\right)}
\newcommand{\transl}[2]{\theta^{-} \left(#1\right)\left(#2\right)}
\newcommand{\transr}[2]{\theta^{+} \left(#1\right)\left(#2\right)}
\newcommand{\sat}[1]{\llbracket #1 \rrbracket}
\newcommand{\tas}[1]{\llparenthesis #1 \rrparenthesis}
\newcommand{\quot}{h}
\newcommand{\upw}[1]{{\uparrow}_{#1}}
\newcommand{\hypo}{\Hypo}
\newcommand{\infer}{\Infer}

\usetikzlibrary{arrows,automata,positioning,calc}
\tikzstyle{WTS}=[every node/.style={font=\fontsize{20}{40}\selectfont}, every state/.style={circle, draw=black, minimum size=15mm}, node distance=1cm, thick, ->, scale=0.5, transform shape]

\forestset{
  edges/.style={double}
}

\theoremstyle{plain} %\crefname{satz}{Satz}{S\"atze}
\def\eg{{\em e.g.}}

\begin{document}

\title[Reasoning About Bounds in Weighted Transition Systems]{Reasoning About Bounds in Weighted Transition Systems\rsuper*}
\titlecomment{{\lsuper*} This paper is an extended version of \cite{hansen2016}.}

\author[M.~Hansen]{Mikkel Hansen}	%required
\address{Department of Computer Science, Aalborg University, Denmark}	%required
\email{\{mhan,kgl,mardare,mrp\}@cs.aau.dk}  %optional
%\thanks{thanks 1, optional.}	%optional

\author[K.G.~Larsen]{Kim Guldstrand Larsen}	%optional
% \address{Department of Computer Science, Aalborg University, Denmark}	%optional
% \email{kgl@cs.aau.dk}  %optional
%\thanks{thanks 2, optional.}	%optional

\author[R.~Mardare]{Radu Mardare}	%optional
% \address{Department of Computer Science, Aalborg University, Denmark}	%optional
% \email{mardare@cs.aau.dk}  %optional
%\thanks{thanks 3, optional.}	%optional

\author[M.R.~Pedersen]{Mathias Ruggaard Pedersen}	%optional
% \address{Department of Computer Science, Aalborg University, Denmark}	%optional
% \email{mrp@cs.aau.dk}  %optional

%% etc.

%% required for running head on odd and even pages, use suitable
%% abbreviations in case of long titles and many authors:

%%%%%%%%%%%%%%%%%%%%%%%%%%%%%%%%%%%%%%%%%%%%%%%%%%%%%%%%%%%%%%%%%%%%%%%%%%%

%% the abstract has to PRECEDE the command \maketitle:
%% be sure not to issue the \maketitle command twice!

%% Abstract
\begin{abstract}
  We propose a way of reasoning about minimal and maximal values of
  the weights of transitions in a weighted transition system (WTS). This perspective induces a notion of bisimulation that is coarser than the classic bisimulation: it
  relates states that exhibit transitions to bisimulation classes with the weights within the same boundaries.
  We propose a customized modal logic that expresses these numeric boundaries for transition weights by means of particular modalities.
  We prove that our logic is invariant under the proposed notion of bisimulation.
  We show that the logic enjoys the finite model property
  and we identify a complete axiomatization for the logic.
  Last but not least, we use a tableau method to show that the satisfiability problem for the logic is decidable.
\end{abstract}

\maketitle

%% Introduction 
\section{Introduction}
Weighted transition systems (WTSs) are used to model concurrent and distributed systems
in the case where some resources are involved, such as time, bandwidth, fuel, or energy consumption.
Recently, the concept of a cyber-physical system (CPS),
which considers the integration of computation and the physical world has become relevant in modeling various real-life situations. 
In these models, sensor feedback affects computation, and through machinery, computation can further affect physical processes.
The quantitative nature of weighted transition systems is well-suited for the quantifiable inputs and sensor measurements of CPSs,
but their rigidity makes them less well-suited for the uncertainty inherent in CPSs.
In practice, there is often some uncertainty attached to the
resource cost, whereas weights in a WTS are precise.
Thus, the model may be too restrictive and unable to capture the uncertainties
inherent in the domain that is being modeled.

In this paper, we attempt to remedy this shortcoming by introducing
a modal logic for WTSs that allows for approximate
reasoning by speaking about upper and lower bounds for the weights of the transitions.
The logic has two types of modal operators that reason about the minimal and maximal weights on transitions, respectively.
This allows reasoning about models where the quantitative information may be imprecise 
(\eg\ due to imprecisions introduced when gathering real data), but where we can establish a lower and upper bound for transitions.

In order to provide the semantics for this logic, we use the set of possible transition weights from
one state to a set of states as an abstraction of the actual transition weights.
The logic is expressive enough to characterize WTSs up to a relaxed notion of weighted bisimilarity,
where the classical conditions are replaced with conditions requiring that the minimal
and maximal weights on transitions are matched.

In \cite{esik2014}, Zolt\'{a}n \'{E}sik also considered the issue of bisimulation for weighted transition systems,
although in the more general setting of synchronization trees with weights in an arbitrary monoid or semiring.
Synchronization trees arise by unfolding the transitions of a weighted transition system starting in some state
which will become the root of the tree.
Both \'{E}sik's and our notion of bisimilarity bears some resemblance
to probabilistic bisimulation \cite{probabilistic_bisimulation},
by considering not only single transitions but transitions to equivalence classes of states.
However, while we require that the upper and lower bounds of these transitions should match,
the bisimilarity of \'{E}sik requires that the sum of the transitions should be the same.
This is motivated by the fact that the synchronization trees do not form a category
which respects the additive structure of a semiring.
However, as \'{E}sik proves, if one takes the quotient with respect to
his version of weighted bisimilarity,
then the category one obtains does respect the additive structure.
Thus, the semiring structure of the weights is of vital importance to \'{E}sik's work,
but is an aspect that we have not considered in our work. 

Our main contribution is a complete axiomatization of our logic,
showing that any validity in this logic can be proved as a theorem from the axiomatic system.
Completeness allows us to transform any validity checking problem into a theorem proving one that can
be solved automatically by modern theorem provers,
thus bridging the gap to the theorem proving community.
The completeness proof adapts the classical filtration method,
which allows one to construct a (canonical) model using maximal consistent sets of formulae.
The main difficulty of adapting this method to our setting
is that we must establish both lower and upper bounds for the transitions in this model.
To achieve this result, we demonstrate that our logic enjoys the finite model property.

Our second significant contribution is a decision procedure for determining the satisfiability of formulae in our logic.
This decision procedure makes use of the tableau method to construct a tableau for a given formula.
If the constructed tableau is successful,
then the formula is satisfiable, and a finite model for the formula can be generated from the tableau.

%% Related work
\subsection*{Related Work.}
Several logics have been proposed in the past to express properties of quantified (weighted, probabilistic or stochastic) systems.
They typically use modalities indexed with real numbers to express properties such as \textit{``$\varphi$ holds with at least probability $b$''},
\textit{``we can reach a state satisfying $\varphi$ with a cost at least $r$''}, etc.

In the context of weighted automata, weighted monadic second order logic
has been introduced by Droste and Gastin \cite{droste2005}
to capture the behaviour of weighted automata for commutative semirings.
This work has been extended to many closely related systems \cite{babari2016}\cite{droste2006a}\cite{droste2006b}\cite{meinecke2006}\cite{fichtner2011}.
There has also been work on connecting weighted monadic second order logic with probabilistic CTL \cite{bollig2009}.
For weighted transition systems, weighted modal logic has been introduced by Larsen and Mardare \cite{larsen1}
to reason about the consumption of resources in such a system.
This logic has been extended to handle recursion \cite{larsen2014a}\cite{larsen2014b}
as well as parallel composition and concurrency \cite{LarsenMX15}.
For both the original weighted modal logic and its concurrent extension,
complete axiomatizations were developed.
A weighted extension of the $\mu$-calculus was introduced by Larsen et al. in \cite{larsen2015},
where a complete axiomatization for this extension was also given.

While our setting is that of weighted transition systems,
our logic and the development of its theory has more in common with Markovian logic
than with the previously mentioned work on weighted systems.

Markovian logic was introduced by Mardare et al. \cite{mardare2012}\cite{cardelli2011a}
building on previous work on probability logics \cite{Zhou09}\cite{Fagin}\cite{Heifetz200131}.
Markovian logic reasons about probabilistic and stochastic systems
using operators $L_r$ and $M_r$ which mean that a property hold with \emph{at least} probability $r$
or \emph{at most} probability $r$, respectively.
Much of the work on Markovian logic has focused on giving a complete axiomatization for the logic {\cite{KozenMP13},
culminating in a Stone duality for Markov processes \cite{6571564}.
However, compositional aspects have been considered in \cite{cardelli2011b},
where also an axiomatization was given for Markovian logic with an operator for parallel composition.

While our logical syntax resembles that of Markovian logic,
our semantics is different in the sense that we argue not about probabilities,
but about an interval of possible weights.
For instance, in the aforementioned logics we have a validity of type $\vdash\neg L_r\phi\to M_r\phi$
saying that the value of the transition from the current state to $\phi$ is either at least $r$ or at most $r$;
on the other hand, in our logic the formula $\neg L_r\phi\land \lnot M_r\phi$ might have a model since $L_r\phi$ and $M_r\phi$
express the fact that the lower cost of a transition to $\phi$ is at least $r$ and the highest cost is at most $r$ respectively.

Our completeness proof uses a technique similar to the one used for weighted modal logic \cite{larsen1} and Markovian logic \cite{KozenMP13}\cite{mardare2012}\cite{cardelli2011a}.
It is however different from these related constructions since our axiomatization is finitary, while the aforementioned ones require infinitary proof rules.
Our axiomatic systems are related to the ones mentioned above and the mathematical structures revealed by this work are also similar to the related ones.
This suggest a natural extension towards a Stone duality result along the lines of \cite{6571564}, which we will consider in a future work. 

Decidability results regarding satisfiability have also been given for some related logics,
such as weighted modal logic \cite{Larsen2016} and probabilistic versions of CTL and the $\mu$-calculus \cite{katoen:sat}.
However, the satisfiability problem is known to be undecidable for
other related logics, in particular timed logics such as TCTL \cite{ALUR19932}
and timed modal logic \cite{DBLP:journals/entcs/JaziriLMX14}.
This fact suggests that our logic is an interesting one which, despite its expressivity, remains decidable.

Our approach of considering upper and lower bounds is related
to interval-based formalisms such as interval Markov chains (IMCs) \cite{JonssonL91}
and interval weighted modal transition systems (WMTSs) \cite{Juhl2012408}.
Much like our approach, IMCs consider upper and lower bounds on transitions
in the probabilistic case.
WMTSs add intervals of weights to individual transitions
of modal transition systems, in which there can be both may- and must-transitions.
A main focus of the work both on IMCs and WMTSs have been
a process of refinement, making the intervals progressively smaller
until an implementation is obtained.
However, none of these works have explored the logical perspective up to the level of axiomatization or satisfiability results,
which is the focus of our paper.

%% Model
\section{Model}
The models addressed in this paper are weighted transition systems,
in which transitions are labeled with numbers to specify the cost of the corresponding transition.
In order to specify and reason about properties regarding imprecision, such as
``the maximum cost of going to a safe state is $10$''
and ``the minimum cost of going to a halting state is $5$'',
we will abstract away the individual transitions
and only consider the minimum and maximum costs from a state to another.
We will do this by constructing for any two states the set of weights that are allowed
from one to the other.

First we recap the definition of a weighted transition system.
Let $\mathcal{AP}$ be a countable set of atomic propositions.
A WTS is formally defined as follows:
\begin{defi}
  A \emph{weighted transition system (WTS)} is a tuple $\mathcal{M} = (S, \rightarrow, \ell)$, where
  \begin{itemize}
    \item $S$ is a non-empty set of \emph{states},
    \item $\rightarrow \subseteq S \times \mathbb{R}_{\geq 0} \times S$ is the \emph{transition relation}, and
    \item $\ell : S \to 2^{\mathcal{AP}}$ is a \emph{labeling function} mapping to each state a set of atomic propositions.
  \end{itemize}
\end{defi}
Note that we impose no restrictions on the state space $S$; it can be uncountable.
We write $s \xrightarrow{r} t$ to mean that $(s,r,t) \in \rightarrow$.
We will say that a WTS is \emph{image-finite} if for any $s \in S$
there are only finitely many $t \in S$ such that $s \xrightarrow{r} t$
for some $r \in \mathbb{R}_{\geq 0}$.

When modeling cyber-physical systems,
it is often unreasonable to expect one to know the
exact weights for transitions.
However, it is often the case that one has some bounds
on the actual weights, \eg\ one might know
that the cost of taking some transition is between $5$ and $25$.
In order to reason about these bounds, we abstract away
the individual transitions, and instead consider the set
of weights between a state and a set of states.

\begin{defi}\label{def:theta}
  For an arbitrary WTS $\mathcal{M} = (S,\rightarrow,\ell)$, the function
  $\theta_{\mathcal{M}} : S \to \left(2^S \to 2^{\mathbb{R}_{\geq 0}}\right)$
  is defined for any state $s \in S$ and set of states $T \subseteq S$ as
  \[
  \transm{s}{T}
  =
  \{r \in \mathbb{R}_{\geq 0} \mid \exists t \in T \;\mbox{such that}\; s \xrightarrow{r} t\}.
  \]
\end{defi}

Thus $\transm{s}{T}$ is the set of all possible weights of going from $s$ to a state in $T$.
We will sometimes refer to $\trans{s}{T}$ as the \emph{image from $s$ to $T$}
or simply as an \emph{image set}.
In the rest of the paper, we will use the notation
\[\transl{s}{T} = \begin{cases} -\infty & \text{if } \trans{s}{T} = \emptyset \\ \inf\trans{s}{T} & \text{otherwise} \end{cases}\]
and
\[\transr{s}{T} = \begin{cases} \infty & \text{if } \trans{s}{T} = \emptyset \\ \sup\trans{s}{T} & \text{otherwise.} \end{cases}\]
Thus $\transl{s}{T}$ will be a lower bound on the weights from $s$ to $T$
and $\transr{s}{T}$ will be an upper bound.

\begin{exa}
  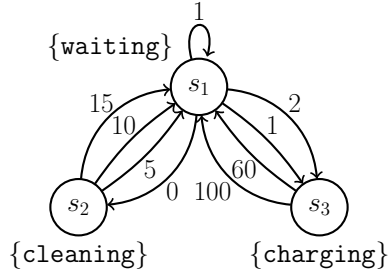
\begin{figure}
    \centering
    \begin{tikzpicture}[WTS, node distance=3cm]
      \node[state, label=above left:{\{\tt waiting\}}] (0) {$s_1$};
      \node[state, label=below:{\{\tt cleaning\}}] (1) [below left = of 0] {$s_2$};
      \node[state, label=below:{\{\tt charging\}}] (2) [below right = of 0] {$s_3$};
      
      \path[->] (0) edge [loop above] node {$1$} (0);
      
      \path[->] (0) edge [bend left = 10, above] node [above, xshift=1.5mm] {$1$} (2);
      \path[->] (0) edge [bend left = 40, above] node [above, xshift=1.5mm] {$2$} (2);
      \path[->] (2) edge [bend left = 10, below] node [below, xshift=-2mm] {$60$} (0);
      \path[->] (2) edge [bend left = 40, below] node [below left, xshift=2mm] {$100$} (0);
      
      \path[->] (0) edge [bend left = 40, below] node [below, xshift=1mm] {$0$} (1);
      \path[->] (1) edge [bend right = 10, below] node [below, xshift=1mm] {$5$} (0);
      \path[->] (1) edge [bend left = 10, above] node [above, xshift=-2mm] {$10$} (0);
      \path[->] (1) edge [bend left = 40, above] node [above, xshift=-2mm] {$15$} (0);
    \end{tikzpicture}
    \caption{A simple model of a robot vacuum cleaner.}
    \label{fig:wts-example}
  \end{figure}
  
  Figure \ref{fig:wts-example} shows a simple model of a robot vacuum cleaner that
  can be in a waiting state, a cleaning state, or a charging state.
  This is an example of a cyber-physical system where
  the costs of transitions are necessarily imprecise.
  The time it takes to recharge the batteries depends on the
  condition of the batteries as well as that of the charger;
  the time it takes to clean the room depends on how dirty the room is,
  and how free the floor is from obstacles;
  and the time it takes to reach the charger depends on where in the room
  the robot is when it needs to be recharged.
  By constructing the image sets, we can abstract away from the individual transitions.
  For example, we have $\trans{s_2}{\{s_1\}} = \{5,10,15\}$,
  so $\transl{s_2}{\{s_1\}} = 5$ and $\transr{s_2}{\{s_1\}} = 15$.
\end{exa}

We will now establish some useful properties of image sets.
In particular, the transition function is monotonic with respect to set inclusion,
and union distributes over image sets as one might expect.

\begin{lem}[Monotonicity of $\theta$]\label{lem:thetamono}
  Let $\mathcal{M} = (S,\rightarrow,\ell)$ be a WTS
  and let $T_1$ and $T_2$ be subsets of $S$.
  If $T_1 \subseteq T_2$, then $\trans{s}{T_1} \subseteq \trans{s}{T_2}$.
\end{lem}

\begin{lem}\label{lem:thetaunion}
  Let $\mathcal{M} = (S,\rightarrow,\ell)$ be a WTS.
  For any $s \in S$ and $T_1,T_2 \subseteq S$, it holds that
  \begin{enumerate}
    \item $\trans{s}{T_1 \cup T_2} = \trans{s}{T_1} \cup \trans{s}{T_2}$ and
    \item $\trans{s}{T_1 \cap T_2} \subseteq \trans{s}{T_1} \cap \trans{s}{T_2}$.
  \end{enumerate}
\end{lem}

As usual we would like some way of relating model states with equivalent behavior.
To this end we define the notion of a bisimulation relation.
The classical notion of a bisimulation relation for weighted transition systems \cite{blackburn},
which we term weighted bisimulation, is defined as follows.

\begin{defi}
  Given a WTS $\mathcal{M} = (S, \rightarrow, \ell)$,
  an equivalence relation $\mathcal{R} \subseteq S \times S$ on $S$
  is called a \emph{weighted bisimulation relation}
  iff for all $s,t \in S$, $s \mathcal{R} t$ implies
  \begin{itemize}
    \item (Atomic harmony) $\ell(s) = \ell(t)$,
    \item (Zig) if $s \xrightarrow{r} s'$ then there exists $t' \in S$ such that $t \xrightarrow{r} t'$ and $s' \mathcal{R} t'$, and
    \item (Zag) if $t \xrightarrow{r} t'$ then there exists $s' \in S$ such that $s \xrightarrow{r} s'$ and $s' \mathcal{R} t'$.
  \end{itemize}
\end{defi}

We say that $s,t \in S$ are weighted bisimilar, written $s \sim_W t$,
iff there exists a weighted bisimulation relation $\mathcal{R}$ such that $s \mathcal{R} t$.
Weighted bisimilarity, $\sim_W$, is the largest weighted bisimulation relation.

Since it is our goal to abstract away from the exact weights on the transitions,
the bisimulation that we will now introduce
does not impose the classical zig-zag conditions \cite{blackburn} of a bisimulation relation,
but instead require that bounds be matched for any bisimulation class.

\begin{defi}\label{def:bisim}
  Given a WTS $\mathcal{M} = (S, \rightarrow, \ell)$,
  an equivalence relation $\mathcal{R} \subseteq S \times S$ on $S$
  is called a \emph{generalized weighted bisimulation relation}
  iff for all $s,t \in S$, $s \mathcal{R} t$ implies
  \begin{itemize}
    \item (Atomic harmony) $\ell(s) = \ell(t)$,
    \item (Lower bound) $\transl{s}{T} = \transl{t}{T}$, and
    \item (Upper bound) $\transr{s}{T} = \transr{t}{T}$
  \end{itemize}
  for any $\mathcal{R}$-equivalence class $T \subseteq S$.
\end{defi}

Given $s,t \in S$ we say that $s$ and $t$
are generalized weighted bisimilar, written $s \sim t$,
iff there exists a generalized weighted bisimulation relation
$\mathcal{R}$ such that $s \mathcal{R} t$.
We let $\sim$ denote generalized weighted bisimilarity
which is defined as
\[
\mathord{\sim} = \bigcup \left\{\mathcal{R} \mid \mathcal{R} \text{ is a generalized weighted bisimulation relation} \right\}.
\]

We will now show that generalized weighted bisimilarity, $\sim$, is the largest generalized weighted bisimulation relation.
To this end, we first need to show that $\sim$ is an equivalence relation.

\begin{lem}\label{lem:simequiv}
  Generalized weighted bisimilarity, $\sim$, is an equivalence relation.
\end{lem}
\begin{proof}
  In order to prove that generalized weighted bisimilarity is an equivalence relation, we have to show that it is reflexive, symmetric and transitive.
  \begin{description}
  \item[Reflexivity] Consider the identity relation
    \[
    \mathcal{I} = \left\{(s,s) \mid s \in S \text{ for some WTS } \mathcal{M} = (S, \rightarrow, \ell) \right\}.
    \]
    It is trivial to verify that $\mathcal{I}$ is a generalized weighted bisimulation relation, and therefore $\mathcal{I} \subseteq \mathord{\sim}$.

  \item[Symmetry] Let $\mathcal{M} = (S, \rightarrow, \ell)$ be a WTS and $s,t \in S$ states such that $s \sim t$.
  Because $s \sim t$ there must exist a generalized weighted bisimulation relation $\mathcal{R}$ such that $s \mathcal{R} t$.
  Since $\mathcal{R}$ is an equivalence relation, we immediately get $t \sim s$.

  \item[Transitivity] Let $\mathcal{M} = (S, \rightarrow, \ell)$ be a WTS and $s,t,u \in S$ states such that $s \sim t$ and $t \sim u$.
  There must exist generalized weighted bisimulation relations $\mathcal{R}$ and $\mathcal{R}'$ such that $s \mathcal{R} t$ and $t \mathcal{R}' u$.
  Let $\mathcal{R}'' = (\mathcal{R} \cup \mathcal{R}')^{+}$ be the transitive closure of the union of $\mathcal{R}$ and $\mathcal{R}'$.
  Since $\mathcal{R}$ and $\mathcal{R}'$ are both equivalence relations, $\mathcal{R} \cup \mathcal{R}'$ is reflexive and symmetric,
  and since the transitive closure of a symmetric and reflexive relation is symmetric and reflexive,
  we get that $\mathcal{R}''$ is an equivalence relation.
  We need to show that $\mathcal{R}''$ is a generalized weighted bisimulation relation.
  Atomic harmony is trivially satisfied.

  Suppose that $\trans{u}{T''} \neq \emptyset$ for some $T'' \in S/\mathcal{R}''$ implying the existence of a state $u' \in T''$ such that $\trans{u}{\{u'\}} \neq \emptyset$,
  further implying the existence of an equivalence class $T' \in S/\mathcal{R}'$ such that $u' \in T'$ and thus $\trans{u}{T'} \neq \emptyset$.
  $t \mathcal{R}' u$ implies $\trans{t}{T'} \neq \emptyset$ which further implies the existence of a state $t' \in T'$ such that $\trans{t}{\{t'\}} \neq \emptyset$.
  There must exist an equivalence class $T \in S/\mathcal{R}$ such that $t' \in T$ implying $\trans{t}{T} \neq \emptyset$.
  Because $s \mathcal{R} t$ we must have $\trans{s}{T} \neq \emptyset$ implying the existence of a state $s' \in T$ such that $\trans{s}{\{s'\}} \neq \emptyset$.
  $s',t' \in T$ implies $s' \mathcal{R} t'$, $t',u' \in T'$ implies $t' \mathcal{R}' u'$,
  and therefore $s' \mathcal{R}'' u'$ implying $s' \in T''$ which further implies $\trans{s}{T''} \neq \emptyset$.
  Therefore $\trans{u}{T''} \neq \emptyset$ implies $\trans{s}{T''} \neq \emptyset$ for all $T'' \in S/\mathcal{R}''$.
  Symmetric arguments show that $\trans{s}{T''} \neq \emptyset$ implies $\trans{u}{T''} \neq \emptyset$ for all $T'' \in S/\mathcal{R}''$,
  and therefore $\trans{s}{T''} = \emptyset$ if and only if $\trans{u}{T''} = \emptyset$ for all $T'' \in S/\mathcal{R}''$.

  Suppose towards a contradiction that $\transl{s}{T''} \neq \transl{u}{T''}$ for some $T'' \in S/\mathcal{R}''$.
  We have two cases to consider, namely $\transl{s}{T''} < \transl{u}{T''}$ and $\transl{s}{T''} > \transl{u}{T''}$.
  If $\transl{s}{T''} < \transl{u}{T''}$ there must exist a rational number $q \in \mathbb{Q}$ such that $\transl{s}{T''} < q < \transl{u}{T''}$,
  implying the existence of a state $s' \in T''$ such that $\transl{s}{T''} \leq \transl{s}{\{s'\}} < q$.
  There must exist $T \in S/\mathcal{R}$ such that $s' \in T$ implying $\transl{s}{T} < q$.
  Because $s \mathcal{R} t$ we must have $\transl{s}{T} = \transl{t}{T}$ implying the existence of a state $t' \in T$ such that $\transl{t}{\{t'\}} < q$.
  There must exist $T' \in S/\mathcal{R}'$ such that $t' \in T'$ implying $\transl{t}{T'} < q$.
  Because $t \mathcal{R}' u$ we must have $\transl{t}{T'} = \transl{u}{T'}$ implying the existence of a state $u' \in T'$ such that $\transl{u}{\{u'\}} < q$.
  $s',t' \in T$ implies $s' \mathcal{R} t'$, $t',u' \in T'$ implies $t' \mathcal{R} u'$,
  and therefore $s' \mathcal{R}'' u'$, implying $u' \in T''$ and therefore $\transl{u}{T''} < q$, leading to a contradiction.
  Symmetric arguments show that also $\transl{s}{T''} > \transl{u}{T''}$ leads to a contradiction and therefore $\transl{s}{T} = \transl{u}{T}$ for any $T \in S/\mathcal{R}''$.

  Similar arguments show that $\transr{s}{T} = \transr{u}{T}$ for any $T \in S/\mathcal{R}''$ thus showing that $\mathcal{R}''$
  is a generalized weighted bisimulation relation implying $\mathcal{R}'' \subseteq \mathord{\sim}$ and therefore $s \sim t$ and $t \sim u$ implies $s \sim u$. \qedhere
  \end{description}
\end{proof}

Having established that $\sim$ is an equivalence relation, we will now show that it is indeed the largest generalized weighted bisimulation relation.

\begin{thm}
  Generalized weighted bisimilarity, $\sim$, is the largest generalized weighted bisimulation relation.
\end{thm}
\begin{proof}
  We first show that $\sim$ is a generalized weighted bisimulation relation.
  By Lemma \ref{lem:simequiv} we know that $\sim$ is an equivalence relation.
  Let $\mathcal{M} = (S, \rightarrow, \ell)$ be a WTS and $s,t \in S$ states such that $s \sim t$.
  There must exist a generalized weighted bisimulation relation $\mathcal{R}$ such that $s \mathcal{R} t$,
  which trivially verifies atomic harmony.
  
  Suppose that $\trans{t}{T} \neq \emptyset$ for some $T \in S/\mathord{\sim}$, implying the existence of a state $t' \in T$ such that $\trans{t}{\{t'\}} \neq \emptyset$.
  There must exist an equivalence class $T' \in S/\mathcal{R}$ such that $t' \in T'$, which implies that $\trans{t}{T'} \neq \emptyset$.
  Because $s \mathcal{R} t$ we must have $\trans{s}{T'} \neq \emptyset$, implying the existence of a state $s' \in T'$ such that $\trans{s}{\{s'\}} \neq \emptyset$.
  Because $s',t' \in T'$ we must have $s' \mathcal{R} t'$ and hence $s' \sim t'$,
  so $s' \in T$ and thus $\trans{s}{T} \neq \emptyset$.
  Symmetric arguments show that $\trans{s}{T} \neq \emptyset$ implies $\trans{t}{T} \neq \emptyset$
  and therefore $\trans{s}{T} = \emptyset$ if and only if $\trans{t}{T} = \emptyset$ for all $T \in S/\mathord{\sim}$.

  Suppose $\transl{s}{T} \neq \transl{t}{T}$ for some $T \in S/\mathord{\sim}$.
  We have two cases to consider, namely $\transl{s}{T} < \transl{t}{T}$ and $\transl{s}{T} > \transl{t}{T}$.
  If $\transl{s}{T} < \transl{t}{T}$ there must exist a rational number $q \in \mathbb{Q}$ such that $\transl{s}{T} < q < \transl{t}{T}$,
  implying the existence of a state $s' \in T$ such that $\transl{s}{T} \leq \transl{s}{\{s'\}} < q$.
  There must exist $T' \in S/\mathcal{R}$ such that $s' \in T'$ and hence $\transl{s}{T'} < q$.
  Because $s \mathcal{R} t$ we have $\transl{s}{T'} = \transl{t}{T'}$, which means that there exists a state $t' \in T'$ such that $\transl{t}{\{t'\}} < q$.
  $s',t' \in T'$ implies $s' \mathcal{R} t'$ which further implies $s' \sim t'$ and therefore $\transl{t}{T} < q$, leading to a contradiction.
  Symmetric arguments show that also $\transl{s}{T} > \transl{t}{T}$ leads to a contradiction, and therefore $\transl{s}{T} = \transl{t}{T}$ for all $T \in S/\mathord{\sim}$.

  Similar arguments show that $\transr{s}{T} = \transr{t}{T}$ for any $T \in S/\mathord{\sim}$, thus showing that $\sim$ is a generalized weighted bisimulation relation.
  
  $\sim$ was defined as the union of all generalized weighted bisimulation relations,
  so for any generalized weighted bisimulation relation $\mathcal{R}$
  we must have $\mathcal{R} \subseteq \mathord{\sim}$,
  and hence we conclude that $\sim$ is the largest generalized weighted bisimulation relation.
\end{proof}

In what follows, we will use bisimulation to mean generalized weighted bisimulation and bisimilarity to mean generalized weighted bisimilarity.

\begin{exa}\label{ex:sim_nwsim}
  Consider the WTS depicted in Figure \ref{fig:ex_bisim_nwbisim}.
  It is easy to see that $\{s',t'\}$ is a $\sim$-equivalence class,
  and in fact it is the only $\sim$-equivalence class with in-going transitions.
  Since $\transl{s}{\{s',t'\}} = \transl{t}{\{s',t'\}} = 1$
  and $\transr{s}{\{s',t'\}} = \transr{t}{\{s',t'\}} = 3$ we must have $s \sim t$,
  but because $s \xrightarrow{2} s'$ and $t \not \xrightarrow{2}$ it cannot be the case that $s \sim_W t$.
  
  \begin{figure}
    \begin{tikzpicture}[WTS, node distance=2cm]
      \node[state, label=left:{$\{a\}$}]  (s0)               {$s$};
      \node[state, label=left:{$\{b\}$}]  (s1) [below=of s0] {$s'$};
      \node[state, label=right:{$\{a\}$}] (t0) [right=of s0] {$t$};
      \node[state, label=right:{$\{b\}$}] (t1) [below=of t0] {$t'$};

      \path (s0) edge[bend right=45] node[left] {$1$} (s1);
      \path (s0) edge                node[left] {$2$} (s1);
      \path (s0) edge[bend left=45]  node[left] {$3$} (s1);

      \path (t0) edge[bend right=45] node[left] {$1$} (t1);
      \path (t0) edge[bend left=45]  node[left] {$3$} (t1);
    \end{tikzpicture}
    \captionof{figure}{$s \sim t$ but $s \not \sim_W t$.}
    \label{fig:ex_bisim_nwbisim}
  \end{figure}
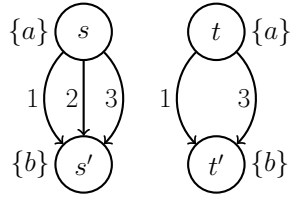
\end{exa}

The following lemma shows that if two states are weighted bisimilar,
then their image sets match exactly for any weighted bisimulation class.

\begin{lem}\label{lem:bisim}
  Let $\mathcal{M} = (S, \rightarrow, \ell)$ be a WTS and let $s,t \in S$.
  $s \sim_W t$ implies that $\trans{s}{T} = \trans{t}{T}$
  for any $\sim_W$-equivalence class $T \subseteq S$.
\end{lem}
\begin{proof}
  Assume $s \sim_W t$ and let $T \subseteq S$ be a $\sim_W$-equivalence class.
  If $r \in \trans{s}{T}$, then there exists some $s' \in T$ such that $s \xrightarrow{r} s'$.
  Because $s \sim_W t$, there must exist some $t' \in T$ such that $t \xrightarrow{r} t'$
  and $s' \sim_W t'$. Since $T$ is a $\sim_W$-equivalence class,
  this means that $r \in \trans{t}{T}$.
  A similar argument shows that if $r \in \trans{t}{T}$,
  then $r \in \trans{s}{T}$.
\end{proof}

We can now show the following relationship between $\sim$ and $\sim_W$.

\begin{thm}\label{thm:bisimcoarse}
  Generalized weighted bisimilarity is coarser than weighted bisimilarity, i.e.
  \[
  \sim_W \mathbin{\subseteq} \sim \quad\text{and}\quad \sim_W \mathbin{\neq} \sim.
  \]
\end{thm}
\begin{proof}%[Proof of Theorem \ref{thm:bisimcoarse}]
  Assume that $s \sim_W t$. We have that $\ell(s) = \ell(t)$,
  and by Lemma \ref{lem:bisim}, we have that $\trans{s}{T} = \trans{t}{T}$
  for any $\sim_W$-equivalence class $T \subseteq S$.
  This implies in particular that
  $\transl{s}{T} = \transl{t}{T}$ and $\transr{s}{T} = \transr{t}{T}$.
  Hence $\sim_W$ is a bisimulation relation.
  
  By Example \ref{ex:sim_nwsim}, the inclusion is strict.
\end{proof}

This result is not surprising,
as our bisimulation relation only looks at the extremes
of the transition weights, whereas weighted bisimulation
looks at all of the transition weights.

%% Logic
\section{Logic}\label{sec:logic}
In this section we introduce a modal logic which is inspired by Markovian logic \cite{mardare2012}.
Our aim is that our logic should be able to capture the notion of bisimilar states as presented in the previous section,
and as such it must be able to reason about the lower and upper bounds on transition weights.

\begin{defi}
  The formulae of the logic $\mathcal{L}$ are induced by the abstract syntax
  \[\mathcal{L}: \quad \varphi, \psi ::= p \mid \neg \varphi \mid \varphi \land \psi \mid L_r \varphi \mid M_r \varphi\]
  where $r \in \mathbb{Q}_{\geq 0}$ is a non-negative rational number and $p \in \mathcal{AP}$ is an atomic proposition.
\end{defi}

$L_r$ and $M_r$ are modal operators. An illustration of how $L_r$ and $M_r$ are interpreted can be seen in Figure \ref{fig:semantics}.
Intuitively, $L_r \varphi$ means that the cost of transitions to where $\varphi$
holds is \emph{at least} $r$ (see Figure \ref{fig:Lsemantics}), and $M_r \varphi$ means that
the cost of transitions to where $\varphi$ holds is \emph{at most} $r$ (see Figure \ref{fig:Msemantics}).
We now give the precise semantics interpreted over WTSs.

\begin{defi}
  Given a WTS $\mathcal{M} = (S, \rightarrow, \ell)$,
  a state $s \in S$ and a formula $\varphi \in \mathcal{L}$,
  the satisfiability relation $\models$ is defined inductively as
  \[
  \begin{array}{l l l}
    \mathcal{M},s \models p                   & \mbox{ iff } & p \in \ell(s), \\
    \mathcal{M},s \models \neg \varphi        & \mbox{ iff } & \mathcal{M},s \not\models \varphi, \\
    \mathcal{M},s \models \varphi \wedge \psi & \mbox{ iff } & \mathcal{M},s \models \varphi \;\text{and}\; \mathcal{M},s \models \psi, \\
    \mathcal{M},s \models L_r \varphi         & \mbox{ iff } & \transl{s}{\sat{\varphi}_{\mathcal{M}}} \geq r,\\
    \mathcal{M},s \models M_r \varphi         & \mbox{ iff } & \transr{s}{\sat{\varphi}_{\mathcal{M}}} \leq r,\\
  \end{array}
  \]
  where $\sat{\varphi}_{\mathcal{M}} = \left\{s \in S \mid \mathcal{M},s \models \varphi \right\}$
  is the set of all states of $\mathcal{M}$ having the property $\varphi$.
  \begin{figure}
    \centering
    \begin{subfigure}{0.45\textwidth}
      \centering
      \begin{tikzpicture}[scale=.5]
        %% X-axis
        \draw [->,thick] (0,0) -- (8,0);

        %% Arc
        \draw [-,semithick] (6,0) arc (0:180:2cm) node[above, xshift=10mm, yshift=10mm] {$\trans{s}{\sat{\varphi}}$};

        %% r
        \draw[shift={(1,0)},-] (0pt,5pt) -- (0pt,-5pt) node[below] {$r$};

        %% Lower bound
        \draw[->] (2,0) ++ (0,-.5) -- (2,0) node[below,at start] {$\theta^{-}$};
        %% Upper bound
        \draw[->] (6,0) ++ (0,-.5) -- (6,0) node[below,at start] {$\theta^{+}$};
      \end{tikzpicture}
      \caption{$\mathcal{M},s \models L_r \varphi$}
      \label{fig:Lsemantics}
    \end{subfigure}
    \begin{subfigure}{0.45\textwidth}
      \centering
      \begin{tikzpicture}[scale=.5]
        %% X-axis
        \draw [->,thick] (0,0) -- (8,0);

        %% Arc
        \draw [-,semithick] (6,0) arc (0:180:2cm) node[above, xshift=10mm, yshift=10mm] {$\trans{s}{\sat{\varphi}}$};

        %% r
        \draw[shift={(7,0)},-] (0pt,5pt) -- (0pt,-5pt) node[below] {$r$};

        %% Lower bound
        \draw[->] (2,0) ++ (0,-.5) -- (2,0) node[below,at start] {$\theta^{-}$};
        %% Upper bound
        \draw[->] (6,0) ++ (0,-.5) -- (6,0) node[below,at start] {$\theta^{+}$};
      \end{tikzpicture}
      \caption{$\mathcal{M},s \models M_r \varphi$}
      \label{fig:Msemantics}
    \end{subfigure}
    \caption{The semantics of $L_r$ and $M_r$. If $\mathcal{M},s \models L_r \varphi$,
             then $r$ is to the left of $\transl{s}{\sat{\varphi}}$, and if $\mathcal{M},s \models M_r \varphi$,
             then $r$ is to the right of $\transr{s}{\sat{\varphi}}$.}
    \label{fig:semantics}
  \end{figure}
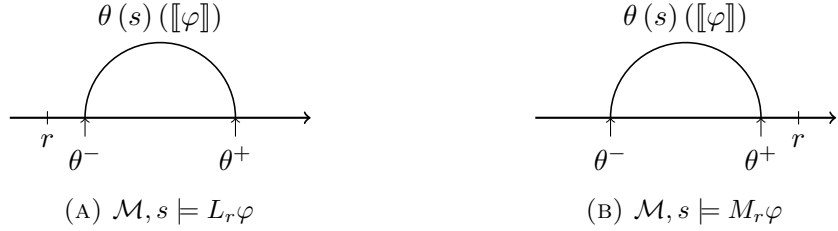
\end{defi}

We will omit the subscript ${\mathcal{M}}$ from $\sat{\varphi}_{\mathcal{M}}$ whenever the model is clear from the context.
If $\mathcal{M},s \models \varphi$ we say that $\mathcal{M}$ is a model of $\varphi$. 
A formula is said to be \emph{satisfiable} if it has at least one model.
We say that $\varphi$ is a \emph{validity} and write $\models \varphi$ if $\neg \varphi$ is not satisfiable.
In addition to the operators defined by the syntax of $\mathcal{L}$,
we also have the derived operators such as $\bot$, $\to$, etc.
defined in the usual way.
A \emph{literal} is a formula that is of the form $p$ or $\neg p$ where $p \in \mathcal{AP}$.

The formula $L_0 \varphi$ has special significance in our logic,
as this formula means that there exists some transition to where $\varphi$ holds.
In fact, it follows in a straightforward manner from the semantics that
$\mathcal{M}, s \models L_0 \varphi$ if and only if $\trans{s}{\sat{\varphi}} \neq \emptyset$.
We can therefore encode the usual box and diamond modalities in our logic in the following way.
\[\Diamond \varphi = L_0 \varphi \quad \Box \varphi = \neg \Diamond \neg \varphi.\]
Notice also that in general, the following schemes \emph{do not hold}.
\begin{align*}
  L_r \varphi \land L_r \psi &\rightarrow L_r(\varphi \land \psi) \\
  M_r \varphi \land M_r \psi &\rightarrow M_r(\varphi \land \psi)
\end{align*}
The reason that they do not hold in general is that there may be no transition to where $\varphi \land \psi$ holds,
i.e. $\neg L_0 (\varphi \land \psi)$. If we assume $L_0 (\varphi \land \psi)$,
then both schemes hold, as we show in Lemma \ref{lem:theorems}.
Another thing to note about the logic is that the formulae $L_r \varphi$ and $L_r \neg \varphi$
can both hold in the same model.
To see this, simply construct a state that
has two transitions with weight $x \geq r$
to two different states,
one where $\varphi$ holds and one where $\varphi$ does not hold.

\begin{exa}\label{ex:logic}
  Consider again our model of a robot vacuum cleaner depicted in Figure \ref{fig:wts-example}.
  Perhaps we want a guarantee that it takes no more than one time unit to go
  from a waiting state to a charging state.
  This can be expressed by the formula ${\tt waiting} \to M_1{\tt charging}$,
  but since we know the only waiting state in our model is $s_1$
  this can be simplified to simply checking whether
  $\mathcal{M},s_1 \models M_1{\tt charging}$.
  We thus have to check that $\transr{s_1}{\sat{{\tt charging}}} \leq 1$.
  We do this by constructing the image set $\trans{s_1}{\sat{{\tt charging}}}$.
  Since $\sat{{\tt charging}} = \{s_3\}$,
  we have $\trans{s_1}{\{s_3\}} = \{1,2\}$.
  Hence $\transr{s_1}{\sat{{\tt charging}}} = 2 \not\leq 1$,
  so $\mathcal{M},s_1 \not\models M_1{\tt charging}$.
\end{exa}

\begin{lem}\label{lem:invariance}
  Let $\mathcal{M} = (S, \rightarrow, \ell)$ be an image-finite WTS and $s \in S$.
  Let $T \subseteq S$ be a set such that all elements of $T$ satisfy exactly the same formulae,
  and furthermore for any $t \in T$ and $t' \notin T$,
  there exists a formula $\varphi$ such that $t \models \varphi$ and $t' \not\models \varphi$.
  Then there exists a formula $\varphi \in \mathcal{L}$ such that
  $\trans{s}{T} = \trans{s}{\sat{\varphi}}$.
\end{lem}
\begin{proof}
  The idea of the proof is to repeatedly use the observation that if $t' \notin T$,
  then there exists a formula $\varphi$ such that $t' \not\models \varphi$ and $t \models \varphi$ for all $t \in T$.
  First pick some formula $\varphi_1$ such that $t \models \varphi_1$ for all $t \in T$.
  Then $T \subseteq \sat{\varphi_1}$, so $\trans{s}{T} \subseteq \trans{s}{\sat{\varphi_1}}$.
  If $\trans{s}{T} \subsetneq \trans{s}{\sat{\varphi_1}}$,
  then there must exist some $t_1 \notin T$
  such that $s \xrightarrow{r} t_1$ and $t_1 \models \varphi_1$.
  Since $t_1 \notin T$, there must exist some formula $\varphi_2$
  such that $t_1 \not \models \varphi_2$ and $t \models \varphi_2$ for all $t \in T$.
  We then get $\trans{s}{T} \subseteq \trans{s}{\sat{\varphi_1 \land \varphi_2}}$.
  Again, if $\trans{s}{T} \subsetneq \trans{s}{\sat{\varphi_1 \land \varphi_2}}$,
  then there must exist some $t_2 \notin T$ such that
  $s \xrightarrow{r} t_2$ and $t_2 \models \varphi_2$.
  Since $t_2 \notin T$, there must exist some formula $\varphi_3$
  such that $t_1 \not \models \varphi_3$ and $t \models \varphi_3$ for all $t \in T$.
  Since $\mathcal{M}$ is image-finite, there can only be finitely many states $t_i \notin T$ with $s \xrightarrow{r} t_i$,
  so continuing in the same way, we will eventually get a formula $\varphi_1 \land \dots \land \varphi_n$
  such that $\trans{s}{T} = \trans{s}{\sat{\varphi_1 \land \dots \land \varphi_n}}$.
\end{proof}

Next we show that our logic $\mathcal{L}$ is invariant under bisimulation,
which is also known as the Hennessy-Milner property.
In order to prove this result, we have to restrict our models to only those that are image-finite, as shown by the following example.

\begin{exa}
  \begin{figure}
    \begin{tikzpicture}[WTS, node distance=2cm]
      % States
      \node[state] (omega)                          {$\omega$};
      \node        (dots1) [below = of omega]       {\Huge $\vdots$};
      \node[state] (n) [below = of dots1]           {$n$};
      \node        (dots2) [below = of n]           {\Huge $\vdots$};
      \node[state] (2) [below = of dots2]           {$2$};
      \node[state] (1) [below = of 2]               {$1$};
      \node[state] (s) [left = 2cm of 2]            {$s$};
      \node[state] (t) [right = 2cm of 2]           {$t$};
      
      \coordinate (c1) at ($(s) + (1,0.8)$);
      \coordinate (c2) at ($(s) + (1.7,-0.7)$);
      \coordinate (c3) at ($(t) + (-1,0.8)$);
      \coordinate (c4) at ($(t) + (-1.7,-0.7)$);
      
      \node (dots3) at (c1) {\Huge $\vdots$};
      \node (dots5) at (c3) {\Huge $\vdots$};

      % Natural numbers
      \path (omega) edge[loop above] node[above] {$0$}  (omega);
      \path (dots1) edge node[left] {$0$}               (n);
      \path (n) edge node[left] {$0$}                   (dots2);
      \path (dots2) edge node[left] {$0$}               (2);
      \path (2) edge node[left] {$0$}                   (1);

      % s
      \path (s) edge[bend left] node[left = 0.1cm] {$2$} (omega);
      
      \path (s) edge[bend left = 10] node[left = 0.1cm] {$1$} (n);
      \path (s) edge[bend right = 10] node[right = 0.1cm] {$4$} (n);
      
      \path (s) edge[bend left = 10] node[above] {$1$} (2);
      \path (s) edge[bend right = 10] node[below] {$4$} (2);
      
      \path (s) edge[bend left = 10] node[right = 0.2cm] {$1$} (1);
      \path (s) edge[bend right = 10] node[left = 0.2cm] {$4$} (1);
      
      % t
      \path (t) edge[bend right] node[right = 0.1cm] {$3$} (omega);
      
      \path (t) edge[bend right = 10] node[right = 0.1cm] {$1$} (n);
      \path (t) edge[bend left = 10] node[left = 0.1cm] {$4$} (n);
      
      \path (t) edge[bend right = 10] node[above] {$1$} (2);
      \path (t) edge[bend left = 10] node[below] {$4$} (2);
      
      \path (t) edge[bend right = 10] node[left = 0.2cm] {$1$} (1);
      \path (t) edge[bend left = 10] node[right = 0.2cm] {$4$} (1);
    \end{tikzpicture}
    \captionof{figure}{$s$ and $t$ satisfy the same logical formulae, but $s \not\sim t$.}
    \label{fig:ex_non_invariance}
  \end{figure}
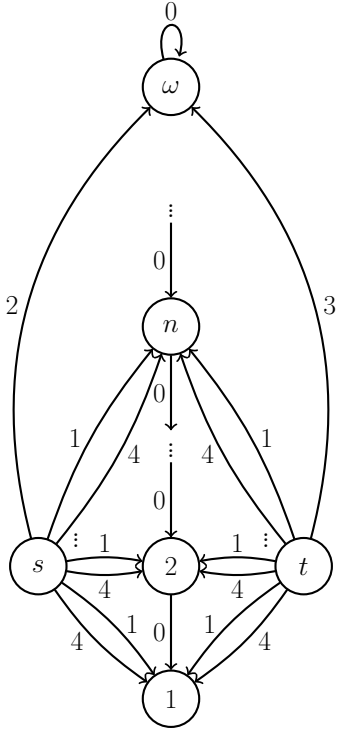
  
  Consider the WTS depicted in Figure \ref{fig:ex_non_invariance}
  with state space $S = \mathbb{N} \cup \{\omega, s, t\}$
  and $\ell(s') = \emptyset$ for all $s' \in S$.
  The transition relation is given by $\omega \xrightarrow{0} \omega$,
  $s \xrightarrow{2} \omega$, $t \xrightarrow{3} \omega$,
  and $n+1 \xrightarrow{0} n$, $s \xrightarrow{1} n$, and $t \xrightarrow{1} n$
  for all $n \in \mathbb{N}$.
  
  Then we have that $s_1 \sim s_2$ if and only if $s_1 = s_2$,
  since any states in $\mathbb{N} \cup \{\omega\}$ can be distinguished by
  the number of steps they can take,
  and $s$ and $t$ can be distinguished by the fact that
  $\transl{s}{\{\omega\}} = 2 \neq 3 = \transl{t}{\{\omega\}}$.
  However, $s$ and $t$ satisfy all the same formulae,
  since any formula that holds in $\omega$ will also hold in $n$ for some $n \in \mathbb{N}$,
  and the weights on the transitions to $\omega$
  will therefore be masked by the bounds $1$ and $4$,
  and hence any formula can not distinguish between $s$ and $t$.
\end{exa}

The proof strategy follows a classical pattern: The left to right direction is shown by induction on $\varphi$ for $\varphi \in \mathcal{L}$.
The right to left direction is shown by constructing a relation $\mathcal{R}$ relating those states that satisfy the same formulae and showing that this relation is a bisimulation relation.

\begin{thm}[Bisimulation invariance]\label{thm:bisiminvariance}
  For any WTS $\mathcal{M} = (S, \rightarrow, \ell)$ and states $s,t \in S$ it holds that
  \[
  s \sim t \quad\mbox{implies}\quad \left[\forall \varphi \in \mathcal{L}.\; \mathcal{M},s \models \varphi \;\;\mbox{iff}\;\; \mathcal{M},t \models \varphi\right].
  \]
  Furthermore, if $\mathcal{M}$ is image-finite, then it also holds that
  \[
  \left[\forall \varphi \in \mathcal{L}.\; \mathcal{M},s \models \varphi \;\;\mbox{iff}\;\; \mathcal{M},t \models \varphi\right] \quad\mbox{implies}\quad s \sim t.
  \]
\end{thm}
\begin{proof}%[Proof of Theorem \ref{thm:bisiminvariance}]
  We first show that $s \sim t$ implies $\mathcal{M},s \models \varphi$ if and only if
  $\mathcal{M},t \models \varphi$ for all $\varphi \in \mathcal{L}$ by induction on $\varphi$.
  The Boolean cases are trivial.
  If $\varphi = L_r \psi$, then
  we have $\transl{s}{\sat{\psi}} \geq r$,
  which implies that $\transl{s}{\sat{\psi}} \neq - \infty$.
  Assume towards a contradiction that $\transl{t}{\sat{\psi}} < r$.
  It can not be the case that $\transl{t}{\sat{\psi}} = - \infty$,
  hence it follows that $\sat{\psi}$ and $\trans{t}{\sat{\psi}}$ are non-empty,
  so there must exist some element $t' \in \sat{\psi}$
  such that $\transl{t}{\sat{\psi}} \leq \transl{t}{\{t'\}} < r$.
  Since $\sim$ is an equivalence relation,
  there must exists some $\sim$-equivalence class $T$ such that $t' \in T$.
  This means that $\{t'\} \subseteq T$,
  so that also $\transl{t}{T} \leq \transl{t}{\{t'\}} < r$.
  By the induction hypothesis we have that $T \subseteq \sat{\psi}$.
  Because $s \sim t$, we have that $\transl{s}{T} = \transl{t}{T} < r$,
  so by monotonicity we get $\transl{s}{\sat{\psi}} \leq \transl{s}{T} < r$,
  which is a contradiction.
  The $M_r$ case is handled similarly.

  For the reverse direction, assume that $\mathcal{M}$ is image-finite.
  We have to show that if for all $\varphi \in \mathcal{L}$,
  $\mathcal{M},s \models \varphi$ if and only if $\mathcal{M},t \models \varphi$ then $s \sim t$.
  To this end, we define a relation $\mathcal{R}$ on $S$ as
  \[
  \mathcal{R} = \left\{(s,t) \in S \times S \mid \forall \varphi \in \mathcal{L}.\; \mathcal{M},s \models \varphi \;\mbox{iff}\; \mathcal{M},t \models \varphi \right\}  .
  \]
  $\mathcal{R}$ is clearly an equivalence relation and $s \mathcal{R} t$.

  It is clear that $\ell(s) = \ell(t)$.
  Next we show that $\transl{s}{T} = \transl{t}{T}$ and $\transr{s}{T} = \transr{t}{T}$ for any $\mathcal{R}$-equivalence class $T$.
  Let $T \subseteq S$ be an $\mathcal{R}$-equivalence class.
  We first show that $\trans{s}{T} = \emptyset$ if and only if $\trans{t}{T} = \emptyset$.
  Assume that $\trans{s}{T} = \emptyset$.
  By Lemma \ref{lem:invariance} there exists a formula $\varphi$
  such that $\trans{s}{T} = \trans{s}{\sat{\varphi}} = \emptyset$,
  and therefore $s \not \models L_0 \varphi$.
  Now assume towards a contradiction that $\trans{t}{T} \neq \emptyset$.
  Since $\mathcal{M}$ is image-finite,
  there must be a finite subset $T' \subseteq T$ such that
  $\trans{t}{T} = \trans{t}{T'}$.
  By Lemma \ref{lem:thetaunion}, we then get
  $\trans{t}{T} = \bigcup_{t' \in T'}\trans{t}{\{t'\}} \neq \emptyset$,
  from which it follows that there must be some $t' \in T'$
  such that $\trans{t}{\{t'\}} \neq \emptyset$.
  Since $t' \in T$, we must have $t' \models \varphi$,
  and therefore $t \models L_0 \varphi$,
  which contradicts the fact that $s \mathcal{R} t$
  and $s \not \models L_0 \varphi$.
  
  Now assume that $\trans{s}{T} \neq \emptyset$ and $\trans{t}{T} \neq \emptyset$.
  We need to show that $\transl{s}{T} = \transl{t}{T}$
  and $\transr{s}{T} = \transr{t}{T}$.
  We do this by contradiction, which gives us four cases to consider:
  $\transl{s}{T} < \transl{t}{T}$, $\transl{s}{T} > \transl{t}{T}$,
  $\transr{s}{T} < \transr{t}{T}$, and $\transr{s}{T} > \transr{t}{T}$.
  
  For the case of $\transl{s}{T} < \transl{t}{T}$,
  there exists $q \in \mathbb{Q}_{\geq 0}$ such that
  \[
  \transl{s}{T} < q < \transl{t}{T} .
  \]
  By Lemma \ref{lem:invariance}, there exists a formula $\varphi$
  such that $\transl{t}{T} = \transl{t}{\sat{\varphi}}$.
  Since $T \subseteq \sat{\varphi}$, we then obtain
  \[
  \transl{s}{\sat{\varphi}} \leq \transl{s}{T} < q < \transl{t}{T} = \transl{t}{\sat{\varphi}} ,
  \]
  which implies that $s \not \models L_q \varphi$ but $t \models L_q \varphi$,
  and thus we get a contradiction.
  The other cases are handled similarly.
\end{proof}

%% Metatheory
\section{Metatheory}
In this section we propose an axiomatization for our logic that we prove not only sound,
but also complete with respect to the proposed semantics.

%% Axiomatic System
\subsection{Axiomatic System}\label{sec:axioms}
Let $r,s \in \mathbb{Q}_{\geq 0}$.
Then the deducibility relation $\vdash \, \subseteq 2^{\mathcal{L}} \times \mathcal{L}$
is a classical conjunctive deducibility relation,
and is defined as the smallest relation which satisfies
the axioms of propositional logic in addition to the axioms given in Table \ref{tab:axioms}.
We will write $\vdash \varphi$ to mean $\emptyset \vdash \varphi$,
and we say that a formula or a set of formulae is \emph{consistent} if it can not derive $\bot$.

\begin{table}
  \centering
  \begin{tabular}{l l l}
    \hline
    (A$1$):   & $\vdash \neg L_0 \bot$ & \\
    (A$2$):   & $\vdash L_{r + q}\varphi \rightarrow L_r \varphi$ & if $q > 0$ \\
    (A$2'$):  & $\vdash M_r\varphi \rightarrow M_{r + q} \varphi$ & if $q > 0$ \\
    (A$3$):   & $\vdash L_r \varphi \land L_q \psi \rightarrow L_{\min\{r,q\}}(\varphi \lor \psi)$ & \\
    (A$3'$):  & $\vdash M_r \varphi \land M_q \psi \rightarrow M_{\max\{r,q\}}(\varphi \lor \psi)$ & \\
    (A$4$):   & $\vdash L_r(\varphi \lor \psi) \rightarrow L_r \varphi \lor L_r \psi$ & \\
    (A$5$):   & $\vdash \neg L_0 \psi \rightarrow (L_r \varphi \rightarrow L_r(\varphi \lor \psi))$ & \\
    (A$5'$):  & $\vdash \neg L_0 \psi \rightarrow (M_r \varphi \rightarrow M_r(\varphi \lor \psi))$ & \\
    (A$6$):   & $\vdash L_{r + q}\varphi \rightarrow \neg M_r\varphi$ & if $q > 0$ \\
    (A$7$):   & $\vdash M_r \varphi \rightarrow L_0 \varphi$ & \\
    (R$1$):   & $\vdash \varphi \rightarrow \psi \implies \vdash (L_r \psi \land L_0 \varphi) \rightarrow L_r \varphi$ & \\
    (R$1'$):  & $\vdash \varphi \rightarrow \psi \implies \vdash (M_r \psi \land L_0 \varphi) \rightarrow M_r \varphi$ & \\
    (R$2$):   & $\vdash \varphi \rightarrow \psi \implies \vdash L_0 \varphi \rightarrow L_0 \psi$ & \\
    \hline
  \end{tabular}
  \caption{The axioms for our axiomatic system, where $\varphi, \psi \in \mathcal{L}$ and $q,r \in \mathbb{Q}$.}
  \label{tab:axioms}
\end{table}

The axioms presented in Table \ref{tab:axioms} bear some resemblance to the axiomatic systems of \cite{mardare2012} and \cite{cardelli2011a}. 
Notably, our axiom A2 is almost identical to A2 of these works and capture similar properties about the systems being studied,
with the major difference being that we reason about transition weights whereas the aforementioned works reason about rates or probabilities of transitions.
Also worth noting here is the similarity between the rule R1 of these works and R1 of our axiomatic system.
A notable difference is that we do not have the additive properties of measures for disjoint sets (since we are not working with probability measures),
as is captured by the axioms A3 and A4 of these works.
Also, in one of the axiomatizations of \cite{mardare2012}, the axioms A2 and A2$'$
are not axioms, but can be derived from the axioms.

Rules R2 and R3 of \cite{mardare2012} and \cite{cardelli2011a} reflect the Archimedean property of rationals,
and while similar axioms can be proven sound in our setting,
these were not needed to show our completeness result.
We suspect, however, that if we were to pursue strong completeness,
infinitary axioms similar to these would be needed.

Axiom A1 captures the notion that since $\bot$ is never satisfied,
we can never take a transition to where $\bot$ holds.
Axiom A2 says that if we know some value is the lower bound for going to
where $\varphi$ holds, then any lower value is also a lower bound for going to where
$\varphi$ holds. Axiom A2$'$ is the analogue for upper bounds.
Axioms A3-A4 show how $L_r$ and $M_r$ distribute over conjunction and disjunction.
The version of axiom A4 where $L_r$ is replaced with $M_r$
is also sound, but as we show in Lemma \ref{lem:theorems},
it can be proven from the other axioms.
Axioms A5 and A5$'$ say that if it is not possible
to take a transition to where $\psi$ holds,
then including the states where $\psi$ holds does not change the bounds.
Axioms A6 and A7 show the relationship between $L_r$ and $M_r$.
In particular, A6 ensures that all bounds are well-formed.
Notice also that the contrapositive of axiom A2 and A7
together gives us that $\neg L_0 \varphi$
implies $\neg L_r \varphi$ and $\neg M_r \varphi$ for any $r \in \mathbb{Q}_{\geq 0}$.
The rules R1 and R1$'$ give a sort of monotonicity for $L_r$ and $M_r$,
and rule R2 says that if $\psi$ follows from $\varphi$,
then if it is possible to take a transition to where $\varphi$ holds,
it is also possible to take a transition to where $\psi$ holds.

We now show some of the theorems which can be deduced from the axioms.
T1, T1$'$, and T5 together complete the distributivity properties for conjunction and disjunction.
T2 and T2$'$ make precise the intuitively clear idea that if two formulae are equivalent,
then their upper and lower bounds should also be the same.
T3 extends axiom A1 to hold for any $r \geq 0$,
and T4 then extends this to any $\varphi$ which implies $\bot$.
\begin{lem}\label{lem:theorems}
  From the axioms listed in Table \ref{tab:axioms} we can derive the following theorems:\\
  \begin{tabular}{l l}
    \emph{(T1):}    & $\vdash (L_r \varphi \land L_q \psi \land L_0(\varphi \land \psi)) \to L_{\max\{r,q\}}(\varphi \land \psi)$ \\
    \emph{(T1$'$):} & $\vdash (M_r \varphi \land M_q \psi \land L_0(\varphi \land \psi)) \to M_{\min\{r,q\}}(\varphi \land \psi)$ \\
    \emph{(T2):}    & $\vdash \varphi \leftrightarrow \psi \implies \vdash L_r \varphi \leftrightarrow L_r \psi$ \\
    \emph{(T2$'$):} & $\vdash \varphi \leftrightarrow \psi \implies \vdash M_r \varphi \leftrightarrow M_r \psi$ \\
    \emph{(T3):}    & $\vdash \neg L_r \bot, \quad r \geq 0$ \\
    \emph{(T4):}    & $\vdash \varphi \to \bot \implies \vdash \neg L_r \varphi,\quad r \geq 0$ \\
    \emph{(T5):}    & $\vdash M_r(\varphi \lor \psi) \rightarrow M_r \varphi \lor M_r \psi$
  \end{tabular}
\end{lem}
\begin{proof}
  \leavevmode
  \begin{description}
  \item[T1]
    Rule R1 implies
    \[
    \vdash \neg L_q (\varphi \land \psi) \rightarrow (\neg L_q \varphi \lor \neg L_0 (\varphi \land \psi))  ,
    \]
    so also
    \[
    \vdash \neg L_q (\varphi \land \psi) \rightarrow (\neg L_q \varphi \lor \neg L_0 (\varphi \land \psi) \lor \neg L_r \psi)  .
    \]
    This is equivalent to
    \[
      \vdash (L_r \varphi \land L_q \psi \land L_0 (\varphi \land \psi)) \rightarrow L_q (\varphi \land \psi)  .
      \]

  \item[T1$'$]
    Similar to T1.

  \item[T2]
    Suppose $\vdash \varphi \leftrightarrow \psi$.
    We have that $\vdash L_r \varphi \rightarrow L_0 \varphi$ by A2
    and $\vdash L_0 \varphi \rightarrow L_0 \psi$ by R2.
    Hence $\vdash L_r \varphi \rightarrow (L_r \varphi \land L_0 \psi)$,
    so $\vdash L_r \varphi \rightarrow L_r \psi$ by R1.
    A similar argument shows that $\vdash L_r \psi \rightarrow L_r \varphi$,
    so $\vdash L_r \varphi \leftrightarrow L_r \psi$.
    
  \item[T2$'$]
    Similar to T2.
    
  \item[T3]
    From axiom A1 we know that $\vdash \neg L_0 \bot$ which,
    by the contrapositive of A2, implies $\vdash \neg L_r \bot$ for any $r > 0$.
    
  \item[T4]
    Suppose $\vdash \varphi \to \bot$.
    We know for any $\psi \in \mathcal{L}$ that $\vdash \bot \to \psi$ and therefore $\vdash \varphi \to \bot \implies \vdash \varphi \leftrightarrow \bot$.
    From A1 we know that $\vdash \neg L_0 \bot$ and from T3 that $\vdash \neg L_r \bot$ for any $r > 0$ implying, by T2, that $\vdash \neg L_r \varphi$ for any $r \geq 0$.
    
  \item[T5]
    By axiom A7 we get $\vdash M_r (\varphi \lor \psi) \rightarrow L_0(\varphi \lor \psi)$
    and A4 gives $\vdash L_0(\varphi \lor \psi) \rightarrow L_0 \varphi \lor L_0 \psi$.
    Hence we get $\vdash M_r(\varphi \lor \psi) \rightarrow (M_r(\varphi \lor \psi) \land L_0 \varphi) \lor (M_r(\varphi \lor \psi) \land L_0 \psi)$.
    Since $\vdash \varphi \rightarrow (\varphi \lor \psi)$
    and $\vdash \psi \rightarrow (\varphi \lor \psi)$,
    rule R1$'$ then gives $\vdash M_r(\varphi \lor \psi) \rightarrow M_r \varphi \lor M_r \psi$. \qedhere
  \end{description}
\end{proof}

Next we prove that our axioms are indeed sound.

\begin{thm}[Soundness]\label{thm:soundness}\hfill
  \[
  \vdash \varphi \quad\text{implies}\quad \models \varphi  .
  \]
\end{thm}
\begin{proof}%[Proof of Theorem \ref{thm:soundness}]
  The soundness of each axiom is easy to show,
  and many of them use the distributive property from Lemma \ref{lem:thetaunion}.
  Here we prove the soundness for a few of the more interesting axioms.
  \begin{description}
  \item[A3]
    Suppose $\mathcal{M},s \models L_r \varphi \land L_q \psi$
    implying that $\mathcal{M},s \models L_r \varphi$ and $\mathcal{M},s \models L_q \psi$,
    implying further that $\transl{s}{\sat{\varphi}} \geq r$ and $\transl{s}{\sat{\psi}} \geq q$.
    
    By Lemma \ref{lem:thetaunion} we must have that 
    \[
    \trans{s}{\sat{\varphi \lor \psi}}
    = \trans{s}{\sat{\varphi} \cup \sat{\psi}}
    = \trans{s}{\sat{\varphi}} \cup \trans{s}{\sat{\psi}}
    \]
    and because $\transl{s}{\sat{\varphi}} \geq r$ and $\transl{s}{\sat{\psi}} \geq q$
    we must have
    \[
    \transl{s}{\sat{\varphi \lor \psi}} = \inf \trans{s}{\sat{\varphi}} \cup \trans{s}{\sat{\psi}}\geq \min\left\{r,q\right\}
    \]
    implying $\mathcal{M},s \models L_{\min\{r,q\}} (\varphi \lor \psi)$.
    
  \item[A4]
    Suppose $\mathcal{M},s \models L_r (\varphi \lor \psi)$ implying that
    \[
    \transl{s}{\sat{\varphi \lor \psi}}
    = \inf \trans{s}{\sat{\varphi}} \cup \trans{s}{\sat{\psi}}
    \geq r  .
    \]
    This implies that at least one of $\trans{s}{\sat{\varphi}}$
    and $\trans{s}{\sat{\psi}}$ is non-empty.
    If $\trans{s}{\sat{\varphi}} \neq \emptyset$, then $\transl{s}{\sat{\varphi}} \geq r$,
    and also if $\trans{s}{\sat{\psi}} \neq \emptyset$, then $\transl{s}{\sat{\psi}} \geq r$,
    so at least one of $\mathcal{M},s \models L_r \varphi$
    and $\mathcal{M},s \models L_r \psi$ must hold.
    Hence $\mathcal{M},s \models L_r \varphi \lor L_r \psi$.
    
  \item[A6]
    Suppose $\mathcal{M},s \models L_{r+q} \varphi$ implying that
    \[
    \transl{s}{\sat{\varphi}} = \inf \trans{s}{\sat{\varphi}} \geq r+q  .
    \]
    It is clear that $\inf \trans{s}{\sat{\varphi}} \leq \sup \trans{s}{\sat{\varphi}}$, so
    \[
    \transr{s}{\sat{\varphi}} = \sup \trans{s}{\sat{\varphi}} \geq \inf \trans{s}{\sat{\varphi}} \geq r + q > r  .
    \]
    Therefore, it cannot be the case that $\mathcal{M},s \models M_r \varphi$
    and thus $\mathcal{M},s \models \neg M_r \varphi$.
    
  \item[R1]
    Suppose $\models \varphi \to \psi$ implying that $\sat{\varphi} \subseteq \sat{\psi}$,
    implying further, by the monotonicity of $\theta$,
    that $\trans{s}{\sat{\varphi}} \subseteq \trans{s}{\sat{\psi}}$. 
    Suppose further that $\mathcal{M},s \models L_r \psi \land L_0 \varphi$
    implying $\mathcal{M},s \models L_r \psi$ and $\mathcal{M},s \models L_0 \varphi$,
    implying further that
    \[
    \transl{s}{\sat{\psi}} = \inf \trans{s}{\sat{\psi}} \geq r \quad \mbox{and} \quad \trans{s}{\sat{\varphi}} \neq \emptyset  .
    \]
    Since $\trans{s}{\sat{\varphi}}$ is non-empty, we then get that
    \[
    \inf \trans{s}{\sat{\varphi}} \geq \inf \trans{s}{\sat{\psi}} \geq r  ,
    \]
    which means that $\mathcal{M},s \models L_r \varphi$. \qedhere
  \end{description}
\end{proof}

%% Finite model property
\subsection{Finite Model Property and Completeness}\label{sec:finitemodel-completeness}
With our axiomatization proven sound we are now ready to present our main results,
namely that our logic has the finite model property and that our axiomatization is complete. 

To show the finite model property we will adapt
the classical filtration method to our setting.
Starting from an arbitrary formula $\rho$,
we define a finite fragment of our logic, $\mathcal{L}[\rho]$,
which we then use to construct a finite model for $\rho$.
The main difference from the classical filtration method
is that we must find an upper and a lower bound for the transitions
in the model.
For an arbitrary formula $\rho \in \mathcal{L}$ we define the following based on $\rho$:
\begin{itemize}
  \item Let $Q_{\rho} \subseteq \mathbb{Q}_{\geq 0}$ be the set of all rational numbers
  $r \in \mathbb{Q}_{\geq 0}$ such that $L_r$ or $M_r$ appears in the syntax of $\rho$.

  \item Let $\Sigma_{\rho}$ be the set of all atomic propositions $p \in \mathcal{AP}$
    such that $p$ appears in the syntax of $\rho$.

  \item The \emph{granularity} of $\rho$, denoted as $gr(\rho)$,
    is the least common denominator of all the elements in $Q_{\rho}$.

  \item The \emph{range} of $\rho$, denoted as $R_{\rho}$, is defined as
      \[R_{\rho} =
      \begin{cases}
        \emptyset & \mbox{if}\; Q_{\rho} = \emptyset\\
        I_\rho \cup \{0\}
        & \mbox{otherwise}  ,
      \end{cases}\]
    where $I_\rho = \left\{
          q \in \mathbb{Q}_{\geq 0} \mid
          \exists j \in \mathbb{N}.\; q = \frac{j}{gr(\rho)} \;\mbox{and}\;
          \min Q_{\rho} \leq q \leq \max Q_{\rho} 
        \right\}$.
    Here the granularity is used to pick out finitely many numbers in the interval.
    Note that we need to add $0$ to $R_{\rho}$
    whether or not $\rho$ actually contains $0$ in any of its modalities.
    This is because, as we have pointed out before, formulae involving $L_0$
    have special significance in our logic.
  \item The \emph{modal depth} of $\rho$, denoted as $md(\rho)$, is defined inductively as:
    \[
    md(\rho) =
    \begin{cases}
      0 & 
      \mbox{if}\; \rho = p \in \mathcal{AP}\\

      md(\varphi) & 
      \mbox{if}\; \rho = \neg \varphi\\

      \max\left\{md(\varphi_1),md(\varphi_2)\right\} & 
      \mbox{if}\; \rho = \varphi_1 \land \varphi_2\\
      1 + md(\varphi) &
      \mbox{if}\; \rho = L_r \varphi \;\text{or}\; \rho = M_r \varphi  .
    \end{cases}
    \]
\end{itemize}
Since all formulae are finite, the modal depth is always a non-negative integer.
The \emph{language} of $\rho$, denoted by $\mathcal{L}[\rho]$, is defined as
\[
\mathcal{L}[\rho] = \{\varphi \in \mathcal{L} \mid R_\varphi \subseteq R_{\rho} , md(\varphi) \leq md(\rho) \;\text{and}\; \Sigma_\varphi \subseteq \Sigma_{\rho} \},
\]
and we take $\mathcal{L}_{\leftrightarrow}[\rho]$ to be the Lindenbaum algebra
of $\mathcal{L}[\rho]$, i.e. the quotient with respect to logical equivalence.
The Lindenbaum algebra is a Boolean algebra with equivalence classes as elements.
Note that the quotient
$\quot : \mathcal{L}[\rho] \rightarrow \mathcal{L}_{\leftrightarrow}[\rho]$
is a homomorphism between Boolean algebras,
and hence preserves the structure of $\mathcal{L}[\rho]$.
For each element $x \in \mathcal{L}_{\leftrightarrow}[\rho]$,
we fix now a formula $\varphi \in x$ to be the representative of that equivalence class,
and we write $\hat{\varphi}$ for $x$.
The order $\leq$ in $\mathcal{L}_{\leftrightarrow}[\rho]$
is then given by $\hat{\varphi} \leq \hat{\psi}$
if and only if $\vdash \varphi \rightarrow \psi$.
The join and meet in $\mathcal{L}_\leftrightarrow[\rho]$ are given by
\[\hat{\varphi} \lor \hat{\psi} = \quot(\varphi \lor \psi) \quad \hat{\varphi} \land \hat{\psi} = \quot(\varphi \land \psi),\]
and complement is given by
\[\neg \hat{\varphi} = \quot(\neg \varphi).\]

Note here the difference between $\quot(\varphi)$ and $\hat{\varphi}$.
The quotient $h$ sends $\varphi$ to its equivalence class $x \in \mathcal{L}_{\leftrightarrow}[\rho]$.
However, it may be the case that $\varphi$ is not the representative for $x$,
but some other formula $\psi$ is. In that case we have $\quot(\varphi) = x = \hat{\psi}$.
On the other hand, $\hat{\varphi}$ denotes both that $\varphi \in \hat{\varphi}$,
and also that $\varphi$ is the chosen representative of its equivalence class,
which ensures that in this case we have $\quot(\varphi) = \hat{\varphi}$.

The idea is that $\Sigma_\rho$ ensures that only finitely many atomic propositions are used,
$R_\rho$ ensures that only finitely many weights on the modalities are used,
and $md(\rho)$ puts a bound on the modal depth of formulae.
The language $\mathcal{L}[\rho]$ itself is not finite,
but contains only finitely many logically non-equivalent formulae.
Hence $\mathcal{L}_{\leftrightarrow}[\rho]$ must be finite,
and as we shall see, it contains all the information necessary to
construct a model for $\rho$.

\begin{prop}
  The language $\mathcal{L}_{\leftrightarrow}[\rho]$ is finite.
\end{prop}
\begin{proof}
  Let $\mathcal{L}_{\leftrightarrow}^n[\rho]$ be the subset of $\mathcal{L}_{\leftrightarrow}[\rho]$
  which only contains formulae of modal depth $n$.
  Then it is clear that
  \[\mathcal{L}_{\leftrightarrow}[\rho] = \bigcup_{i = 0}^{md(\rho)} \mathcal{L}^i_{\leftrightarrow}[\rho].\]
  We will now prove by induction on the modal depth that for each $i$,
  $\mathcal{L}_{\leftrightarrow}^i[\rho]$ is finite.
  
  $i = 0$: In this case, each element of $\mathcal{L}_{\leftrightarrow}^0[\rho]$
    is a Boolean combination of atomic propositions in $\Sigma_\rho$.
    There are $2^{2^{|\Sigma_\rho|}}$ non-equivalent such formulae,
    so this set is finite.
  
  $i > 0$: Each element of $\mathcal{L}_{\leftrightarrow}^i[\rho]$
    is a Boolean combination of formulae of the form $L_r \varphi$ and $M_r \varphi$,
    where $\varphi \in \mathcal{L}_\leftrightarrow^{j}[\rho]$ for some $j < i$ and $r \in R_\rho$.
    By induction hypothesis, we know that there are only finitely many such $\varphi$.
    We know from Lemma \ref{lem:theorems} that if $\varphi$ and $\psi$ are logically equivalent,
    then $L_r \varphi$ and $L_r \psi$ as well as $M_r \varphi$ and $M_r \psi$ are also logically equivalent.
    Since $R_\rho$ is finite, we conclude that $\mathcal{L}_\leftrightarrow^i[\rho]$ is finite.
\end{proof}

In order to define the model,
we need the standard notions of filters and ultrafilters on Boolean algebras \cite{halmos2009}.
A non-empty subset of a Boolean algebra $B$ is called a \emph{filter}
if it is upward-closed with respect to the order,
and closed under finite meets.
A filter $F$ is \emph{proper} if $F \neq B$.
An \emph{ultrafilter} is a proper filter which is maximal
in the sense of set inclusion.

The following property of ultrafilters is often useful.

\begin{lem}
  For an ultrafilter $F$ of $\mathcal{L}_{\leftrightarrow}[\rho]$ it holds that for any
  $\varphi \in \mathcal{L}[\rho]$,
  either $\quot(\varphi) \in F$ or $\neg \quot(\varphi) \in F$,
  but not both.
\end{lem}

We let $\mathcal{U}[\rho]$ denote the set of all ultrafilters on $\mathcal{L}_{\leftrightarrow}[\rho]$.
Since $\mathcal{L}_{\leftrightarrow}[\rho]$ is finite, $\mathcal{U}[\rho]$
is also finite and consequently, any ultrafilter $u \in \mathcal{U}[\rho]$
must be a finite set.
For any set $\Phi \subseteq \mathcal{L}_{\leftrightarrow}[\rho]$,
the characteristic formula of $\Phi$, denoted $\tas{\Phi}$,
is defined as
\[\tas{\Phi} = \bigwedge_{\hat{\varphi} \in \Phi} \varphi.\]
Note that $\tas{\Phi} \in \mathcal{L}[\rho]$ is a finite formula,
and that if $u \in \mathcal{U}[\rho]$, then $\quot(\tas{u}) \in u$.

We will now construct a (finite) model, $\mathcal{M}_{\rho}$,
for $\rho$ with state space $\mathcal{U}[\rho]$.
In order to define the transition relation
$\rightarrow_\rho \subseteq \mathcal{U}[\rho] \times \mathbb{R}_{\geq 0} \times \mathcal{U}[\rho]$,
we consider any two ultrafilters $u,v \in \mathcal{U}[\rho]$ and define two functions $L,M : \mathcal{U}[\rho] \times \mathcal{U}[\rho] \to 2^{R_{\rho}}$ as
\[
L(u,v) = \{r \mid \quot(L_r \tas{v}) \in u\} 
\quad\mbox{and}\quad
M(u,v) = \{s \mid \quot(M_s \tas{v}) \in u\}  .
\]

The following lemma establishes a relationship between $L$ and $M$,
that we will need to define the transition relation.
The lemma is a straightforward consequence of axiom $A7$.

\begin{lem}\label{lem:emptynonempty}
  Given any ultrafilters $u,v \in \mathcal{U}[\rho]$,
  it can not be the case that $L(u,v) = \emptyset$ and $M(u,v) \neq \emptyset$.
\end{lem}
\begin{proof}%[Proof of Lemma \ref{lem:emptynonempty}]
  Assume towards a contradiction that $L(u,v) = \emptyset$
  and $M(u,v) \neq \emptyset$. Then we have $\quot(\neg L_0 \tas{v}) \in u$
  and there exists some $r \in Q_\rho$ such that $\quot(M_r \tas{v}) \in u$.
  However, by axiom A7, this implies that $\quot(L_0 \tas{v}) \in u$,
  which is a contradiction.
\end{proof}

We can now define the transition relation in terms of $L(u,v)$ and $M(u,v)$.
In Figure \ref{fig:finitemodel}, we have illustrated the different cases that we must consider.
Here, the area between $\min Q_\rho$ and $\max Q_\rho$
is the only part that the restricted language $\mathcal{L}[\rho]$ can speak about.
The arches represent the interval within which transitions with that weight are possible.
For any of the arches in the figure, we have the following correspondence with $L_r$ and $M_r$.
\begin{itemize}
  \item If a number $r$ on the real line is contained within the arch,
    then we have $\quot(\neg L_r \tas{v}) \in u$ and $\quot(\neg M_r \tas{v}) \in u$.
  \item If a number $r$ on the real line is to the left of the arch,
    then we have $\quot(L_r \tas{v}) \in u$ and $\quot(\neg M_r \tas{v}) \in u$.
  \item If a number $r$ on the real line is to the right of the arch,
    then we have $\quot(M_r \tas{v}) \in u$ and $\quot(\neg L_r \tas{v}) \in u$.
\end{itemize}
In case (a) in Figure \ref{fig:finitemodel},
we therefore have $L(u,v) \neq \emptyset$ and $M(u,v) \neq \emptyset$,
so we have all the information we need to define the transition.
In case (b) and (f), we have $L(u,v) \neq \emptyset$ and $M(u,v) = \emptyset$,
since there exist numbers within the interval $[\min Q_\rho, \max Q_\rho]$
that are to the left of these arches, but none that are to the right.
This means that we have enough information to define the minimum transition,
but we do not know what the maximum transition is.
Note that we can not simply say that the maximum transition is $\max Q_\rho$,
because that would imply $\quot(M_{\max Q_\rho} \tas{v}) \in u$,
but we know that $M(u,v) = \emptyset$.
Hence we need to pick a number that is to the right of $\max Q_\rho$ as the maximum.
In case (d), we have both $L(u,v) = \emptyset$ and $M(u,v) = \emptyset$.
This implies that $\quot(\neg L_0 \tas{v}) \in u$,
which means that there should be no transition from $u$ to $v$.
In case (c) and (e), we have $L(u,v) = \emptyset$ and $M(u,v) \neq \emptyset$,
but according to Lemma \ref{lem:emptynonempty} these cases can never occur.

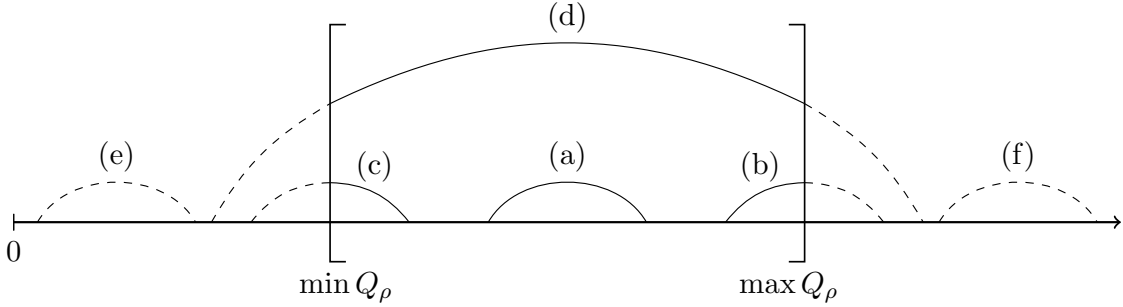
\begin{figure}
  \centering
  \resizebox{\textwidth}{!}{
    \begin{tikzpicture}
      %% X-axis
      \draw [->,thick] (-3,0) -- (11,0);
      %% Zero
      \draw[shift={(-3,0)}] (0pt,3pt) -- (0pt,-3pt);
      \draw[shift={(-3,0)}] (0pt,3pt) -- (0pt,-3pt) node[below] {$0$};
      
      %% Lower bound
      \draw[semithick] (1.2, -0.5) node[below] {$\min Q_\rho$} |- (1, -0.5) -- (1, 2.5) |- (1.2, 2.5);
      %% Upper bound
      \draw[semithick] (6.8, -0.5) node[below] {$\max Q_\rho$} |- (7, -0.5) -- (7, 2.5) |- (6.8, 2.5);

      %% Case a
      \coordinate (1) at (3,0);
      \coordinate (2) at (5,0);
      \draw (1) to [bend left=60] node[above] {(a)} (2);
      
      %% Case b
      \coordinate (3) at (6,0);
      \coordinate (4) at (7,0.5);
      \coordinate (5) at (8,0);
      \draw (3) to [bend left=25] node[above] {(b)} (4);
      \draw[dashed] (4) to [bend left=25] (5);
      
      %% Case c
      \coordinate (6) at (0,0);
      \coordinate (7) at (1,0.5);
      \coordinate (8) at (2,0);
      \draw[dashed] (6) to [bend left=25] (7);
      \draw (7) to [bend left=25] node[above] {(c)} (8);
      
      %% Case d
      \coordinate (9) at (-0.5,0);
      \coordinate (10) at (8.5,0);
      \coordinate (15) at (1,1.5);
      \coordinate (16) at (7,1.5);
      \draw[dashed] (9) to [bend left=15] (15);
      \draw (15) to [bend left=26] node[above] {(d)} (16);
      \draw[dashed] (16) to [bend left=15] (10);
      
      %% Case e
      \coordinate (11) at (-2.7,0);
      \coordinate (12) at (-0.7,0);
      \draw[dashed] (11) to [bend left=60] node[above] {(e)} (12);
      
      %% Case f
      \coordinate (13) at (8.7,0);
      \coordinate (14) at (10.7,0);
      \draw[dashed] (13) to [bend left=60] node[above] {(f)} (14);
    \end{tikzpicture}
  }
  \caption{When constructing a transition from $u$ to $v$, we will only have information about what happens in the region $Q_\rho$ and at $0$. The line represents the non-negative real line and the arches represent the transitions that would be possible in a full model (i.e. one not restricted to $\mathcal{L}[\rho]$). The dashed part of the arches represent the part of the transition that we do not have information about.}
  \label{fig:finitemodel}
\end{figure}

We therefore distinguish the following three cases in order to define the transition relation:
\begin{enumerate}
  \item If $L(u,v) \neq \emptyset$ and $M(u,v) \neq \emptyset$,
    then we add the two transitions $u \xrightarrow{r_1} v$ and $u \xrightarrow{r_2} v$
    where $r_1 = \max L(u,v)$ and $r_2 = \min M(u,v)$. \label{item:trans1}
  \item If $L(u,v) \neq \emptyset$ and $M(u,v) = \emptyset$,
    then we add the two transitions $u \xrightarrow{r_1} v$ and $u \xrightarrow{r_2} v$
    where $r_1 = \max L(u,v)$ and $r_2 = \max Q_\rho + \frac{1}{gr(\rho)}$. \label{item:trans2}
  \item If $L(u,v) = \emptyset$ and $M(u,v) = \emptyset$,
    then there is no transition from $u$ to $v$. \label{item:trans3}
\end{enumerate}
The following lemma tells us that these transitions are well-formed,
i.e. that the lower bound on transitions is less than or equal to the upper bound.

\begin{lem}\label{lem:leq}
  For any ultrafilters $u,v \in \mathcal{U}[\rho]$,
  if $L(u,v) \neq \emptyset$ and $M(u,v) \neq \emptyset$,
  then $\max L(u,v) \leq \min M(u,v)$.
\end{lem}
\begin{proof}
  Assume towards a contradiction that $\max L(u,v) > \min M(u,v)$.
  Then there exist $q,q' \in Q_\rho$ such that $q > q'$,
  $\quot(L_q \tas{v}) \in u$ and $\quot(M_{q'} \tas{v}) \in u$.
  Since $q > q'$, axiom A6 gives $\quot(\neg M_{q'} \tas{v}) \in u$,
  which is a contradiction.
\end{proof}

Finally we define the labeling function $\ell_\rho : \mathcal{U}[\rho] \rightarrow 2^{\mathcal{AP}}$
for any $u \in \mathcal{U}[\rho]$ as $ \ell_\rho(u) = \{p \in \mathcal{AP} \mid p \in u\} $.
We then have a model $\mathcal{M}_\rho = (\mathcal{U}[\rho], \rightarrow_\rho, \ell_\rho)$, and it is not difficult to prove that $\mathcal{M}_\rho$ is a WTS.
Before we can prove the truth lemma,
we need the following technical lemma.

\begin{lem}\label{lem:filtration}
  For any consistent formula $\varphi \in \mathcal{L}[\rho]$,
  if $[\mathcal{M}_\rho,u \models \varphi$ iff $\quot(\varphi) \in u]$, then
  \[
  \bigvee_{v \in \sat{\varphi}} \quot(\tas{v}) \in u \quad \text{iff} \quad \quot(\varphi) \in u.
  \]
\end{lem}
\begin{proof}
  Suppose $\bigvee_{v \in \sat{\varphi}} \quot(\tas{v}) \in u$.
  Assume towards a contradiction that $\quot(\neg \tas{v}) \in u$ for all $v \in \sat{\varphi}$.
  Then, since $u$ is an ultrafilter,
  we must have $\bigwedge_{v \in \sat{\varphi}} \quot(\neg \tas{v}) \in u$,
  which means that $\neg \bigvee_{v \in \sat{\varphi}} \quot(\tas{v}) \in u$,
  which is a contradiction.
  Hence there exists some $v' \in \sat{\varphi}$ such that $\quot(\tas{v'}) \in u$.
  If $\hat{\psi} \in v'$, then $\vdash \tas{v'} \rightarrow \psi$,
  so $\hat{\psi} \in u$ because $u$ is an ultrafilter.
  Since $v' \in \sat{\varphi}$, we have by assumption that $\quot(\varphi) \in v'$,
  so we get $\quot(\varphi) \in u$.
  
  Suppose $\quot(\varphi) \in u$, which by assumption means that $u \in \sat{\varphi}$,
  so $\vdash \tas{u} \rightarrow \bigvee_{v \in \sat{\varphi}} \tas{v}$.
  Since $u$ is an ultrafilter, we have $\quot(\tas{u}) \in u$,
  and hence $\bigvee_{v \in \sat{\varphi}} \quot(\tas{v}) \in u$.
\end{proof}

We are now in a position to state and prove the truth lemma,
which says that an ultrafilter satisfies a formula in our model
if and only if that formula is included in the ultrafilter.

\begin{lem}[Truth lemma]\label{lem:truth}
  If $\rho \in \mathcal{L}$ is a consistent formula,
  then for all $\varphi \in \mathcal{L}[\rho]$
  and $u \in \mathcal{U}[\rho]$ we have
  \[\mathcal{M}_\rho,u \models \varphi \quad \text{iff} \quad \quot(\varphi) \in u.\]
\end{lem}
\begin{proof}%[Proof of Lemma \ref{lem:truth}]
  The proof is by induction on the structure of $\varphi$.
  The Boolean cases are trivial.
  For the case $\varphi = L_r \psi$, we proceed as follows.

    ($\implies$)
      Assume $\mathcal{M}_\rho, u \models L_r \psi$,
      meaning that $\transl{u}{\sat{\psi}} \geq r$.
      It can not be the case that $\trans{u}{\sat{\psi}} = \emptyset$,
      because otherwise $\transl{u}{\sat{\psi}} = - \infty$,
      and we have assumed $\transl{u}{\sat{\psi}} \geq r$.
      It also can not be the case that $\sat{\psi} = \emptyset$,
      because otherwise $\trans{u}{\sat{\psi}} = \emptyset$.
      We can partition all the ultrafilters $v \in \sat{\psi}$
      as follows.
      Let $E = \{v \in \sat{\psi} \mid L(u,v) = \emptyset\}$
      and $N = \{v \in \sat{\psi} \mid L(u,v) \neq \emptyset\}$.
      We then get that $E \cap N = \emptyset$, $E \cup N = \sat{\psi}$,
      $\quot(\neg L_0 \tas{v}) \in u$ for all $v \in E$,
      and $\quot(L_r \tas{v}) \in u$ for all $v \in N$.
      Because $u$ is an ultrafilter, we then have
      \[
      \quot\left(\bigwedge_{v \in E} \neg L_0 \tas{v} \land \bigwedge_{v \in N} L_r \tas{v}\right) \in u  .
      \]
      By axiom A3, this implies
      \[
      \quot\left(\bigwedge_{v \in E} \neg L_0 \tas{v} \land L_r \bigvee_{v \in N} \tas{v}\right) \in u  .
      \]
      Then axiom A5 gives
      \[
      \quot\left(L_r \bigvee_{v \in \sat{\psi}} \tas{v}\right) \in u .
      \]
      By the induction hypothesis, T2, and Lemma \ref{lem:filtration},
      we then get $\quot(L_r \psi) \in u$.
    
    ($\impliedby$)
      Let $\quot(L_r \psi) \in u$. It follows from A1, A2, and R2 that $\psi$ is consistent. 
      Hence, by the induction hypothesis, $\sat{\psi}$ is non-empty.
      We first show that $\trans{u}{\sat{\psi}} \neq \emptyset$.
      Assume therefore towards a contradiction that $\trans{u}{\sat{\psi}} = \emptyset$.
      Then for all $v \in \sat{\psi}$, we must have that case \ref{item:trans3} holds,
      and hence $L(u,v) = \emptyset$,
      meaning $\quot(\neg L_r \tas{v}) \in u$ for all $v \in \sat{\psi}$.
      Since there are finitely many $v \in \sat{\psi}$,
      we can enumerate them as $v_1,v_2,\dots,v_n$.
      Then, since $u$ is an ultrafilter, we have
      \[
      \quot\left(\neg L_r \tas{v_1} \land \neg L_r \tas{v_2} \land \dots \land \neg L_r \tas{v_n}\right) \in u  .
      \]
      By De Morgan's law, this is equivalent to
      \[
      \quot\left(\neg (L_r \tas{v_1} \lor L_r \tas{v_2} \lor \dots \lor L_r \tas{v_n})\right) \in u  .
      \]
      The contrapositive of axiom A4 then gives that
      \[
      \quot\left(\neg L_r (\tas{v_1} \lor \tas{v_2} \lor \dots \lor \tas{v_n})\right) \in u  ,
      \]
      and by the induction hypothesis, T2, and Lemma \ref{lem:filtration},
      this is equivalent to $\neg \quot(L_r \psi) \in u$, which is a contradiction.
      
      Now assume towards a contradiction that $\transl{u}{\sat{\psi}} < r$. Then
      there exists some $v \in \sat{\psi}$ such that $\transl{u}{\{v\}} < r$ and
      case \ref{item:trans1} or case \ref{item:trans2} holds.
      In either case we have $\max L(u,v) < r$ and hence
      there exists some $q \in Q_\rho$ such that
      $\quot(L_q \tas{v}) \in u$, which implies $\quot(L_0 \tas{v}) \in u$ by axiom A2.
      By the induction hypothesis, $\quot(\psi) \in v$,
      which means that $\vdash \tas{v} \rightarrow \psi$.
      rule R1 then gives $\quot(L_r \tas{v}) \in u$,
      but this is a contradiction since $\max L(u,v) < r$.
    
    The $M_r$ case is similar, using axiom A7 instead of A2 to derive $\quot(L_0 \psi) \in u$.
\end{proof}

Having established the truth lemma, we can now show that any consistent formula is satisfied by some finite model.

\begin{thm}[Finite model property]\label{thm:finitemodel}
  For any consistent formula $\varphi \in \mathcal{L}$,
  there exists a finite WTS $\mathcal{M} = (S, \rightarrow, \ell)$ and a state $s \in S$
  such that $\mathcal{M},s \models \varphi$.
\end{thm}
\begin{proof}%[Proof of Theorem \ref{thm:finitemodel}]
  Since $\varphi \in \mathcal{L}$ is consistent,
  $\quot(\varphi) \neq \quot(\bot)$,
  and since $\mathcal{L}_{\leftrightarrow}[\rho]$ is finite,
  there must exist an ultrafilter $u \in \mathcal{U}[\rho]$ such that $\quot(\varphi) \in u$.
  By the truth lemma, this means that $\mathcal{M}_\varphi, u \models \varphi$,
  and by construction, $\mathcal{M}_\varphi$ is a finite model.
\end{proof}

We are now able to state and prove our main result, namely that our axiomatization is complete.

\begin{thm}[Completeness]\label{thm:completeness}
  For any formula $\varphi \in \mathcal{L}$, it holds that
  \[
  \models \varphi \quad \text{implies} \quad \vdash \varphi  .
  \]
\end{thm}
\begin{proof}%[Proof of Theorem \ref{thm:completeness}]\hfill
  \[
  \models \varphi \quad \text{implies} \quad \vdash \varphi
  \]
  is equivalent to
  \[
  \not \vdash \varphi \quad \text{implies} \quad \not \models \varphi  ,
  \]
  which is equivalent to
  \[
  \mbox{the consistency of}\; \neg \varphi \;\mbox{implies the existence of a model for}\; \neg \varphi  ,
  \]
  and this is guaranteed by the finite model property.
\end{proof}

We have thus established completeness for our logic.
There is also a stronger notion of completeness, often called strong completeness,
which asserts that $\Phi \models \varphi$ implies $\Phi \vdash \varphi$
for any set of formulae $\Phi \subseteq \mathcal{L}$.
Completeness is a special case of strong completeness where $\Phi = \emptyset$.
In the case of compact logics, strong completeness follows directly from completeness.
However, our logic is non-compact.

\begin{thm}\label{thm:noncompact}
  Our logic is non-compact, meaning that there exists an infinite set $\Phi \subseteq \mathcal{L}$
  such that each finite subset of $\Phi$ admits a model, but $\Phi$ does not.
\end{thm}
\begin{proof}
  Consider the set
  $
  \Phi = \{L_q \varphi \mid q < r \} \cup \{\neg L_r \varphi\}
  $.
  For any finite subset of $\Phi$, it is easy to construct a model.
  However, if $\mathcal{M},s \models L_q \varphi$ for all $q < r$
  where $q,r \in \mathbb{Q}_{\geq 0}$,
  then by the Archimedean property of the rationals,
  we also have $\mathcal{M},s \models L_r \varphi$.
  Hence there can be no model for $\Phi$. 
\end{proof}

%% Satisfiability
\section{Satisfiability}\label{sec:sat}
The finite model property gives us a way of
deciding in general whether there exists a model which satisfies a given formula.
An algorithm would be to enumerate all finite models
and all theorems derivable from the axioms,
which can be done since there are countably many of each of these.
If $\varphi$ is satisfiable, it has a model,
and by the finite model property, it has a finite one.
So we can check one by one whether a finite model satisfies $\varphi$.
On the other hand, if $\varphi$ is not satisfiable,
then $\neg \varphi$ is a theorem,
so we can search through all theorems to see whether $\neg \varphi$ is one of them.
Since $\varphi$ is either satisfiable or its negation is a theorem,
one of these two algorithms must eventually halt.
By running these two algorithms in parallel,
we have shown that the problem of deciding satisfiability for a given formula is decidable.

In what follows we do more:
We propose an algorithm that constructs a tableau syntactically from a given formula.
By inspecting this tableau, we can decide whether or not the formula is satisfiable,
and if it is satisfiable, we can construct a model for the formula from the tableau.

As in the previous section, we impose an order on formulae
given by $\varphi \leq \psi$ if and only if $\models \varphi \rightarrow \psi$.
Given a finite set of formulae $\Gamma = \{\varphi_1, \dots, \varphi_n\}$,
we denote by $\min(\Gamma)$ the set of minimal elements of $\Gamma$, i.e.
\[\min(\Gamma) = \{\varphi_i \in \Gamma \mid \text{there is no } \varphi_j \text{ such that } \varphi_j \leq \varphi_i\},\]
and we let
\[\mathcal{L}(\Gamma) = \{\varphi_i \in \Gamma \mid \text{there is no } j < i \text{ such that } \models \varphi_j \leftrightarrow \varphi_i\}.\]
Furthermore, we let $\upw{\Gamma}(\varphi)$ be the upward closure of $\varphi$ in $\Gamma$, i.e.
\[\upw{\Gamma}(\varphi) = \{\varphi' \in \Gamma \mid \varphi \leq \varphi'\}.\]

\begin{table}
  \begin{tabular}{c c}
    \hline
    
    & \\
    
    {\begin{prooftree}
      \hypo{\langle \Gamma \cup \{\varphi \land \psi\}, \mathcal{I}^L, \mathcal{I}^M \rangle}
      \infer[left label = {($\land$)}]1{\langle \Gamma \cup \{\varphi, \psi\}, \mathcal{I}^L, \mathcal{I}^M \rangle}
    \end{prooftree}}
    
    &
    
    {\begin{prooftree}
      \hypo{\langle \Gamma \cup \{\neg (\varphi \land \psi)\}, \mathcal{I}^L, \mathcal{I}^M \rangle}
      \infer[left label = {($\neg \land$)}]1{\langle \Gamma \cup \{\neg \varphi\}, \mathcal{I}^L, \mathcal{I}^M \rangle \quad \langle \Gamma \cup \{\neg \psi\}, \mathcal{I}^L, \mathcal{I}^M \rangle}
    \end{prooftree}} \\
    
    & \\
  
    \multicolumn{2}{c}{
      {\begin{prooftree}
        \hypo{\langle \Gamma \cup \{\neg \neg \varphi\}, \mathcal{I}^L, \mathcal{I}^M \rangle}
        \infer[left label = {($\neg\neg$)}]1{\langle \Gamma \cup \{\varphi\}, \mathcal{I}^L, \mathcal{I}^M \rangle}
      \end{prooftree}}
    } \\
    
    & \\
    
    \multicolumn{2}{c}{
      {\begin{prooftree}
        \hypo{\langle \Gamma \cup \{N^1_{r_1} \varphi_1, \dots, N^n_{r_n} \varphi_n\} \cup \{\neg O^1_{r_1'} \varphi_1', \dots, \neg O^{n'}_{r_{n'}'} \varphi_{n'}'\}, \mathcal{I}^L, \mathcal{I}^M \rangle}
        \infer[left label = {(mod)}]1{\langle \{\psi_1\}, \mathcal{I}^L_1, \mathcal{I}^M_1 \rangle \quad \cdots \quad \langle \{\psi_k\}, \mathcal{I}^L_k, \mathcal{I}^M_k \rangle}
      \end{prooftree}}
    } \\
    
    & \\
    
    \multicolumn{2}{c}{
      \parbox{10cm}{
        if $N^i \in \{L,M\}$ for all $1 \leq i \leq n$, $O^j \in \{L,M\}$ for all $1 \leq j \leq n'$,
        and no formula in $\Gamma$ is of the form $N_r \varphi$ or $\neg N_r \varphi$ where $N \in \{L,M\}$.
      }
    } \\
    
    & \\
    
    \hline
  \end{tabular}
  \caption{Tableau rules}
  \label{tab:rules}
\end{table}

A \emph{tableau} is a tree with nodes of the form $\langle \Gamma, \mathcal{I}^L, \mathcal{I}^M \rangle$
that is constructed from the rules of Table \ref{tab:rules},
where the (mod) rule may only be used when no other rule can be used.
For each node $\langle \Gamma, \mathcal{I}^L, \mathcal{I}^M \rangle$,
$\Gamma$ is a set of formulae, and $\mathcal{I}^L$ and $\mathcal{I}^M$
are intervals of the form $\lbag a, b \rbag$ where
$a \in \mathbb{R}_{\geq 0} \cup \{-\infty\}$, $b \in \mathbb{R}_{\geq 0} \cup \{\infty\}$, $\lbag \in \{[,(\}$, and $\rbag \in \{], )\}$,
subject to the constraint that $\lbag = ($ if $a = -\infty$ and $\rbag = \; )$ if $b = \infty$.
We will say that an interval $\lbag a, b \rbag$ is \emph{consistent} if $a < b$ or $a = b$ and the interval is closed.

For the rule (mod), the objects $\psi_i$, $\mathcal{I}^L_i$ and $\mathcal{I}^M_i$ in the conclusion are constructed as follows.
The $\psi_i$ are given by
\[\{\psi_1, \dots, \psi_k\} = \min(\mathcal{L}(\{\varphi_1, \dots, \varphi_n\})).\]
We will show later how to actually compute $\{\psi_1, \dots, \psi_k\}$.
Let $\Gamma' = \{\varphi_1, \dots, \varphi_n\}$ and
\[\mathbb{L}^+_i = \{r \mid L_r \varphi_j = N^j_{r_j} \varphi_j \text{ for some } j \text{ and } \varphi_j \in \upw{\Gamma'}(\psi_i)\}\]
\[\mathbb{M}^+_i = \{r \mid M_r \varphi_j = N^j_{r_j} \varphi_j \text{ for some } j \text{ and } \varphi_j \in \upw{\Gamma'}(\psi_i)\}\]
as well as
\[\mathbb{L}^-_i = \{r \mid L_r \varphi_j' = O^j_{r_j} \varphi_j' \text{ for some } j \text{ and } \models \psi_i \rightarrow \varphi_j'\}\]
\[\mathbb{M}^-_i = \{r \mid M_r \varphi_j' = O^j_{r_j} \varphi_j' \text{ for some } j \text{ and } \models \psi_i \rightarrow \varphi_j'\}.\]

Then the intervals $\mathcal{I}^L_i$ and $\mathcal{I}^M_i$ are given by
\[\mathcal{I}^L_i = \begin{cases}
                      [ \max \mathbb{L}^+_i, \min \mathbb{L}^-_i ) & \text{if } \mathbb{L}^+_i \neq \emptyset \text{ and } \mathbb{L}^-_i \neq \emptyset \\
                      [0, \min \mathbb{L}^-_i) & \text{if } \mathbb{L}^+_i = \emptyset \text{ and } \mathbb{L}^-_i \neq \emptyset \\
                      [ \max \mathbb{L}^+_i, \infty) & \text{if } \mathbb{L}^+_i \neq \emptyset \text{ and } \mathbb{L}^-_i = \emptyset \\
                      [0, \infty) & \text{if } \mathbb{L}^+_i = \emptyset \text{ and } \mathbb{L}^-_i = \emptyset
                    \end{cases}\]
\[\mathcal{I}^M_i = \begin{cases}
                      ( \max \mathbb{M}^-_i, \min \mathbb{M}^+_i ] & \text{if } \mathbb{M}^-_i \neq \emptyset \text{ and } \mathbb{M}^+_i \neq \emptyset \\
                      [0, \min \mathbb{M}^+_i ] & \text{if } \mathbb{M}^-_i = \emptyset \text{ and } \mathbb{M}^+_i \neq \emptyset \\
                      ( \max \mathbb{M}^-_i, \infty) & \text{if } \mathbb{M}^-_i \neq \emptyset \text{ and } \mathbb{M}^+_i = \emptyset \\
                      [0, \infty) & \text{if } \mathbb{M}^-_i = \emptyset \text{ and } \mathbb{M}^+_i = \emptyset
                    \end{cases}\]

Informally, one should think of a node $m = \langle \Gamma, \mathcal{I}^L, \mathcal{I}^M \rangle$ as satisfying all the formulas in $\Gamma$.
Moreover, the (mod)-rule signifies a state transition,
where the new states are given by the nodes in the conclusion,
and any transition to $m$ must have a minimum weight that lies in the interval $\mathcal{I}^L$,
and a maximum weight that lies in the interval $\mathcal{I}^M$.

\begin{exa}
  We now illustrate the use of the (mod) rule through an example.
  Consider the node $m = \langle \{p_1, p_2, L_2 p_1, L_4(p_1 \land p_2), L_0 p_3, \neg L_5 p_2, \neg M_6 p_3\}, \mathcal{I}^L, \mathcal{I}^M \rangle$.
  We group the formulas as
  \[\Gamma = \{p_1, p_2\}, \Gamma' = \{L_2 p_1, L_4 (p_1 \land p_2), L_0 p_3\}, \text{ and } \Gamma'' = \{\neg L_5 p_2, \neg M_6 p_3\},\]
  so that $m = \langle \Gamma \cup \Gamma' \cup \Gamma'', \mathcal{I}^L, \mathcal{I}^M \rangle$.
  Since $\Gamma$ only includes literals, it is clear that we can use no other rules,
  so we are allowed to use (mod) on $m$.
  
  We see that $\models (p_1 \land p_2) \rightarrow p_1$,
  and hence $\{\psi_1, \psi_2\} = \{p_1 \land p_2, p_3\}$,
  so there are two children of $m$. For the first child, we find
  \begin{align*}
    &\mathbb{L}^+_1 = \{2,4\} &&\mathbb{M}^+_1 = \emptyset \\
    &\mathbb{L}^-_1 = \{5\} &&\mathbb{M}^-_1 = \emptyset,
  \end{align*}
  and for the second child we find
  \begin{align*}
    &\mathbb{L}^+_2 = \{0\} &&\mathbb{M}^+_2 = \emptyset \\
    &\mathbb{L}^-_2 = \emptyset &&\mathbb{M}^-_2 = \{6\}.
  \end{align*}
  Hence the intervals become
  \begin{align*}
    &\mathcal{I}^L_1 = [4,5) &&\mathcal{I}^M_1 = [0,\infty) \\
    &\mathcal{I}^L_2 = [0,\infty) &&\mathcal{I}^M_2 = (6,\infty),
  \end{align*}
  and our application of the rule becomes
  \[
    \begin{prooftree}
      \hypo{\langle \{p_1, p_2, L_2 p_1, L_4(p_1 \land p_2), L_0 p_3, \neg L_5 p_2, \neg M_6 p_3\}, \mathcal{I}^L, \mathcal{I}^M \rangle}
      \infer[left label = {(mod)}]1{\langle \{p_1 \land p_2\}, [4,5), [0,\infty) \rangle \quad \langle \{p_3\}, [0,\infty) (6,\infty) \rangle}
    \end{prooftree}
  \]
\end{exa}

Given a formula $\varphi$, we will say that a tableau $\mathcal{T}$ is a \emph{tableau for $\varphi$}
if $\langle \{\varphi\}, [0,0], [0,0] \rangle$ is the root of $\mathcal{T}$.

\begin{defi}
  A node $m$ in a tableau is called
  \begin{itemize}
    \item a \emph{modal node} if the (mod)-rule was applied to $m$ and
    \item a \emph{terminal node} if it is either a modal node or a leaf node.
  \end{itemize}
\end{defi}

\begin{defi}\label{def:consistent}
  A node $m = \langle \Gamma, \lbag_1 a,b \rbag_1, \lbag_2 c,d \rbag_2 \rangle$ is \emph{consistent} if
  \begin{itemize}
    \item for any $p \in \mathcal{AP}$ we do not have both $p \in \Gamma$ and $\neg p \in \Gamma$,
    \item $\lbag_1 a,b \rbag_1$ and $\lbag_2 c,d \rbag_2$ are consistent, and
    \item either $a < d$ or $a = d$, $\lbag_1 = [$, and $\rbag_2 = \; ]$.
  \end{itemize}
\end{defi}

\begin{defi}\label{def:success}
  A tableau $\mathcal{T}$ is \emph{successful} if there exists a subtree $\mathcal{T}'$ of $\mathcal{T}$ such that
  \begin{itemize}
    \item every leaf in $\mathcal{T}'$ is also a leaf in $\mathcal{T}$,
    \item if a modal node $m$ is included in $\mathcal{T}'$,
      then every child of $m$ is also included in $\mathcal{T}'$, and
    \item every terminal node in $\mathcal{T}'$ is consistent.
  \end{itemize}
\end{defi}

Given a successful tableau $\mathcal{T}$,
we construct the WTS $\mathcal{M}(\mathcal{T})$ with state $s_{\mathcal{T}}$
using Algorithm~\ref{alg:model}.

\begin{algorithm}
  \SetAlgoLined
    Let $\mathcal{T}'$ be a witness for the fact that $\mathcal{T}$ is successful \;
    $S := \{s_\mathcal{T}\}, \rightarrow := \emptyset, \ell := \emptyset$ \;
    Let $X$ be a stack and $X := \emptyset$ \;
    $X.push((s_\mathcal{T},r))$ where $r$ is the root of $\mathcal{T}'$ \;
    
    \While{$X \neq \emptyset$}{
      $(s,m) := X.pop$ \;
      Let $m = \langle \Gamma, \Delta, (a,b) \rangle$ \;
      \If{$m$ is not a terminal node}{
        Let $m'$ be the left-most child of $m$ in $\mathcal{T}'$ \;
        $X.push((s,m'))$ \;
      }
      \If{$m$ is a leaf node}{
        $\ell := \ell \cup \{(s,p) \mid p \in \mathcal{AP} \text{ and } p \in \Gamma\}$ \;
      }
      \If{$m$ is a modal node}{
        $\ell := \ell \cup \{(s,p) \mid p \in \mathcal{AP} \text{ and } p \in \Gamma\}$ \;
        Let $m_1 = \langle \Gamma_1, \mathcal{I}^L_1, \mathcal{I}^M_1 \rangle, \dots, m_n = \langle \Gamma_n, \mathcal{I}^L_n, \mathcal{I}^M_n \rangle$ be the children of $m$ in $\mathcal{T}'$ \;
        \For{$i = 1, \dots, n$}{
          Let $\mathcal{I}^L_i = \lbag a_i,b_i \rbag$ and $\mathcal{I}^M_i = \lbag c_i,d_i \rbag$ \;
          $x_i := a_i$ \;
          $y_i := \begin{cases} \max\{a_i, \frac{d_i - c_i}{2} + c_i\} & \text{if } d_i \neq \infty \\ \max\{a_i, c_i + 1\} & \text{if } d_i = \infty\end{cases}$ \;
          $S := S \cup \{s_i\}$ \;
          $\rightarrow := \rightarrow \cup \{(s,x_i,s_i), (s,y_i,s_i)\}$ \;
          $X.push((s_i,m_i))$ \;
        }
      }
    }
    
    $\mathcal{M}(\mathcal{T}) := (S, \rightarrow, \ell)$ \;
    
    \Return $(\mathcal{M}(\mathcal{T}), s_{\mathcal{T}})$ \;
  \caption{Constructing the model $\mathcal{M}(\mathcal{T})$ for a successful tableau $\mathcal{T}$.}
  \label{alg:model}
\end{algorithm}

\begin{lem}\label{lem:tableaumodel}
  If $\mathcal{T}$ is a successful tableau for $\varphi$, then $\mathcal{M}(\mathcal{T}),s_{\mathcal{T}} \models \varphi$.
\end{lem}
\begin{proof}
  Let $Y$ be the set of all pairs $(s,m)$ that are added to the stack $X$ by Algorithm \ref{alg:model} at some point during the construction of $\mathcal{M}(\mathcal{T})$.
  We wish to prove that for any $(s,\langle \Gamma, \mathcal{I}^L, \mathcal{I}^M\rangle) \in Y$
  we have $\mathcal{M}(\mathcal{T}),s \models \Gamma$,
  where we write $\mathcal{M}(\mathcal{T}),s \models \Gamma$ to mean $\mathcal{M}(\mathcal{T}),s \models \varphi$ for all $\varphi \in \Gamma$.
  Note that if we can prove this, then it follows that $\mathcal{M}(\mathcal{T}),s_\mathcal{T} \models \varphi$
  since $(s_\mathcal{T}, \langle \{\varphi\}, \emptyset, (0,0) \rangle) \in Y$.
  
  Let $(s, m)$ be an arbitrary element of $Y$ and let $l$ be the length of the longest path from $m$ to a leaf.
  We will prove, by induction on $l$, that $\mathcal{M}(\mathcal{T}), s \models \Gamma$ where $m = \langle \Gamma, \mathcal{I}^L, \mathcal{I}^M \rangle$.
  
  $l = 0$: In this case, $m$ is a leaf. Hence $\Gamma$ only contains literals,
    and by construction we have $p \in \ell(s)$ if and only if $p \in \Gamma$.
    Since $m$ is consistent, we thus get $\mathcal{M}(\mathcal{T}),s \models \Gamma$.
    
  $l > 0$: In this case we consider the different rules that may be applied to $m$.
  \begin{description}
    \item[($\land$)] We have 
      \[
        \begin{prooftree}
          \hypo{m = \langle \Gamma \cup \{\varphi_1 \land \varphi_2\}, \mathcal{I}^L, \mathcal{I}^M \rangle}
          \infer[left label = {($\land$)}]1{m' = \langle \Gamma \cup \{\varphi_1, \varphi_2\}, \mathcal{I}^L, \mathcal{I}^M \rangle}
        \end{prooftree}
      \]
      By induction hypothesis we get $\mathcal{M}(\mathcal{T}),s \models \Gamma \cup \{\varphi_1,\varphi_2\}$.
      This implies that $\mathcal{M}(\mathcal{T}),s \models \varphi_1$ and $\mathcal{M}(\mathcal{T}),s \models \varphi_2$,
      so $\mathcal{M}(\mathcal{T}),s \models \Gamma \cup \{\varphi_1 \land \varphi_2\}$.
      
    \item[($\neg \land$)] We have
      \[
        \begin{prooftree}
          \hypo{m = \langle \Gamma \cup \{\neg (\varphi_1 \land \varphi_2)\}, \mathcal{I}^L, \mathcal{I}^M \rangle}
          \infer[left label = {($\neg \land$)}]1{m_1 = \langle \Gamma \cup \{\neg \varphi_1\}, \mathcal{I}^L, \mathcal{I}^M \rangle \quad m_2 = \langle \Gamma \cup \{\neg \varphi_2\}, \mathcal{I}^L, \mathcal{I}^M \rangle}
        \end{prooftree}
      \]
      We have three cases to consider;
      either $m_1$ is included in $\mathcal{T}'$,
      $m_2$ is included in $\mathcal{T}'$,
      or both $m_1$ and $m_2$ are included in $\mathcal{T}'$.
      If $m_1$ is included in $\mathcal{T}'$ we get, by the induction hypothesis,
      that $\mathcal{M}(\mathcal{T}),s \models \Gamma \cup \{\neg \varphi_1 \}$
      implying that $\mathcal{M}(\mathcal{T}),s \not \models \varphi_1$.
      If $m_2$ is included in $\mathcal{T}'$ we get, by the induction hypothesis,
      that $\mathcal{M}(\mathcal{T}),s \models \Gamma \cup \{ \neg \varphi_2 \}$
      implying that $\mathcal{M}(\mathcal{T}),s \not \models \varphi_2$.
      In either case we get that $\mathcal{M}(\mathcal{T}),s \not \models \varphi_1 \land \varphi_2$ and $\mathcal{M}(\mathcal{T}),s \models \Gamma$,
      and therefore $\mathcal{M}(\mathcal{T}),s \models \Gamma \cup \{\neg (\varphi_1 \land \varphi_2)\}$.
      The last case follows trivially from the preceding arguments.
    
    \item[($\neg\neg$)] We have
      \[
        \begin{prooftree}
          \hypo{m = \langle \Gamma \cup \{\neg \neg \varphi'\}, \mathcal{I}^L, \mathcal{I}^M \rangle}
          \infer[left label = {($\neg\neg$)}]1{m' = \langle \Gamma \cup \{\varphi'\}, \mathcal{I}^L, \mathcal{I}^M \rangle}
        \end{prooftree}
      \]
      By induction hypothesis we know that $\mathcal{M}(\mathcal{T}),s \models \Gamma \cup \{\varphi'\}$,
      so $\mathcal{M}(\mathcal{T}),s \models \Gamma \cup \{\neg \neg \varphi'\}$.
      
    \item[(mod)] We have
      \[
        \begin{prooftree}
          \hypo{m = \langle \Gamma \cup \{N^1_{r_1} \varphi_1, \dots, N^n_{r_n} \varphi_n\} \cup \{\neg O^1_{r_1'} \varphi_1', \dots, \neg O^{n'}_{r_{n'}'} \varphi_{n'}'\}, \mathcal{I}^L, \mathcal{I}^M \rangle}
          \infer[left label = {(mod)}]1{m_1 = \langle \{\psi_1\}, \mathcal{I}^L_1, \mathcal{I}^M_1 \rangle \quad \cdots \quad m_k = \langle \{\psi_k\}, \mathcal{I}^L_k, \mathcal{I}^M_k \rangle}
        \end{prooftree}
      \]
      $\Gamma$ must consist only of literals, because otherwise the (mod) rule could not be used.
      As in the case for $l = 0$, we then get $\mathcal{M}(\mathcal{T}),s \models \Gamma$ since $m$ is consistent.
      Let $\Psi = \{\psi_1, \ldots, \psi_k\}$,
      and for any $1 \leq j \leq k$, let $\mathcal{I}^L_j = \lbag a_j,b_j \rbag$
      and $\mathcal{I}^M_j = \lbag c_j, d_k \rbag$.
      By the induction hypothesis, we know that $\mathcal{M}(\mathcal{T}), s_j \models \psi_j$ for all $j \in \{1,\ldots,k\}$,
      and, by construction, $s_j$ is the only successor of $s$ that satisfies $\psi_j$.    
      Now consider a formula $N^i_{r_i} \varphi_i$.
      There must exist a subset $\Psi_{\varphi_i} \subseteq \Psi$ such that
      $\trans{s}{\sat{\varphi_i}} = \trans{s}{\bigcup_{\psi' \in \Psi_{\varphi_i}} \sat{\psi'}}$.   
      We first consider the case where $N^i = L$.
      Because $\Psi_{\varphi_i}$ is finite,
      there exists $\psi_j' \in \Psi_{\varphi_i}$ such that
      $\transl{s}{\sat{\varphi_i}} = \transl{s}{\sat{\psi_j'}}$,
      implying the existence of $\psi_j \in \Psi$ such that $\transl{s}{\sat{\varphi_i}} = \transl{s}{\sat{\psi_j}} = a_j$.
      We must have $a_j \geq r_i$ implying $\transl{s}{\sat{\varphi_i}} \geq r_i$, and thus $\mathcal{M}(\mathcal{T}),s \models L_{r_i} \varphi_i$.
      In the case where $N^i = M$ we can, similarly to the previous case,
      find $\psi_j \in \Psi$ such that
      $\transr{s}{\sat{\varphi_i}} = \transr{s}{\sat{\psi_j}}$,
      and we know that $d_i \neq \infty$ implying
      \[\transr{s}{\sat{\psi_j}} = \max\left\{ a_j, \frac{d_j - c_j}{2} + c_j \right\} \leq d_j \leq r_i .\]
      Therefore, $\transr{s}{\sat{\varphi}} \leq r_i$ and thus $\mathcal{M}(\mathcal{T}),s \models M_{r_i} \varphi_i$.
      
      Lastly we consider a formula $\neg O^i_{r_i'}\varphi_i'$.
      If there is no $\psi_j \in \Psi$ such that $\models \psi_j \to \varphi_i'$, then,
      by the construction of $\mathcal{M}(\mathcal{T})$, there is no successor
      $s'$ of $s$ such that $\mathcal{M}(\mathcal{T}),s \models \varphi_i'$.
      Therefore, $\transl{s}{\sat{\varphi_i'}} = \infty$ and $\transr{s}{\sat{\varphi_i'}} = - \infty$,
      and thus $\mathcal{M}(\mathcal{T}),s \models \neg O^i_{r_i'}\varphi_i'$
      is trivially satisfied for $O^i \in \{L, M\}$.
      Suppose $\models \psi_j' \to \varphi_i'$ for some $\psi_j' \in \Psi$.
      We first consider the case where $O^i = L$.
      There must exist $\psi_j \in \Psi$ such that
      $\transl{s}{\sat{\varphi_i'}} = \transl{s}{\sat{\psi_j}} = a_j$.
      By the assumption that $\mathcal{T}$ is successful,
      we must have that $m_j$ is consistent.
      Therefore, $a_j < b_j \leq r_{i'}$ implying
      $\transl{s}{\sat{\varphi_i'}} < r_i'$,
      and thus $\mathcal{M}(\mathcal{T}),s \models \neg L_{r_i'}\varphi_i'$.
      In the case where $O^i = M$ we must be able to find $\psi_j \in \Psi$
      such that $\transr{s}{\sat{\varphi_i'}} = \transr{s}{\sat{\psi_j}}$.
      We have to consider $d_j = \infty$ and $d_j \neq \infty$ separately.
      If $d_j = \infty$ we have
      \[\transr{s}{\sat{\psi_j}} = \max\left\{a_j, c_j + 1\right\} > c_j \geq r_i' .\]
      If $d_j \neq \infty$ we have
      \[\transr{s}{\sat{\psi_j}} = \max\left\{a_j, \frac{d_j - c_j}{2} + c_j\right\} > c_j \geq r_i' .\]
      In either case we have that $\trans{s}{\sat{\varphi_i'}} > r_i'$ and
      therefore $\mathcal{M}(\mathcal{T}),s \models \neg M_{r_i'}\varphi_i'$. \qedhere
  \end{description}
\end{proof}

\begin{lem}\label{lem:tableaux}
  Let $\mathcal{T}_1$ and $\mathcal{T}_2$ be tableaux for $\varphi$.
  Then it holds that $\mathcal{T}_1$ is successful if and only if $\mathcal{T}_2$ is successful.
\end{lem}
\begin{proof}
  Assume that $\mathcal{T}_1$ is a successful tableau.
  Let $\mathcal{T}_1'$ be a subtree of $\mathcal{T}_1$
  which witnesses the fact that $\mathcal{T}_1$ is successful.
  If $\mathcal{T}_1'$ is also a subtree of $\mathcal{T}_2$,
  then we are done. If not, let $d$ be the smallest number such that
  $\mathcal{T}_1'$ differs at depth $d$ from any subtree of $\mathcal{T}_2$ with the same root as $\mathcal{T}_2$.
  Note that we must have $d > 0$ because $\mathcal{T}_1'$ and $\mathcal{T}_2$ have the same root.
  Denote by $\mathcal{T}_1' |_n$ the restriction of $\mathcal{T}_1'$ to depth $n$.
  Then $\mathcal{T}_1' |_{d-1}$ is a subtree of $\mathcal{T}_2$.
  
  At this point we note that $\mathcal{T}_1$ and $\mathcal{T}_2$
  contain the same terminal nodes.
  To see this, let the level $k$ terminal nodes be those terminal nodes that can be reached from the root
  by going through $k-1$ terminal nodes.
  We now argue that the level $k$ terminal nodes of $\mathcal{T}_1$ and $\mathcal{T}_2$
  are the same by induction on $k$.
  
  $k = 1$:
  The level $1$ terminal nodes of $\mathcal{T}_1$ and $\mathcal{T}_2$
  must be the same, since they are all constructed by applying the $(\land)$, $(\neg\land)$, or $(\neg\neg)$
  rules to the root node $\langle \{\varphi\}, [0,0], [0,0] \rangle$.
  
  $k > 1$:
  Since the level $k-1$ terminal nodes of $\mathcal{T}_1$ and $\mathcal{T}_2$
  are the same, they must also have the same children,
  which are constructed from the (mod) rule.
  Hence each level $k$ terminal node is constructed by applying the $(\land)$, $(\neg\land)$, or $(\neg\neg)$
  rules to a child of one of the level $k-1$ terminal nodes,
  so they are also the same in $\mathcal{T}_1$ and $\mathcal{T}_2$.
  
  Now let $X$ be the set of all terminal nodes that are in $\mathcal{T}_1'$ at depth $d$ or below.
  Since every node in $X$ is a node in $\mathcal{T}_1$,
  it must also be a node in $\mathcal{T}_2$.
  Furthermore, every node in $X$ is reachable in $\mathcal{T}_2$ from $\mathcal{T}_1' |_{d-1}$.
  Hence, if we extend $\mathcal{T}_1' |_{d-1}$ to include all paths in $\mathcal{T}_2$ leading from the leaves of $\mathcal{T}_1' |_{d-1}$ to an element in $X$,
  then this extension is a subtree of $\mathcal{T}_2$.
  Denote this extension by $\mathcal{T}_2'$.
  
  Finally we argue that $\mathcal{T}_2'$ is a witness for the fact that $\mathcal{T}_2$ is successful
  by checking the three conditions of Definition~\ref{def:success}.
  Every leaf of $\mathcal{T}_2'$ is also a leaf in $\mathcal{T}_2$,
  since all the leaves of $\mathcal{T}_2'$ are elements of $X$.
  This takes care of the first condition.
  If $m_i$ is a child of the modal node $m$ in $\mathcal{T}_2$,
  and $m$ is included in $\mathcal{T}_2'$,
  then $m$ is also a modal node in $\mathcal{T}_1'$,
  and hence $m_i$ must be included in $\mathcal{T}_1'$.
  This means that there is a terminal node $m'_1$ in $\mathcal{T}_1'$
  which is reached by $m_i$.
  Hence, if $m_i$ is not included in $\mathcal{T}_2'$,
  then the terminal node $m'_1$ can not be reached in $\mathcal{T}_2'$,
  but this contradicts how $\mathcal{T}_2'$ was constructed.
  Therefore $m_i$ must also be included in $\mathcal{T}_2'$,
  so the second condition is satisfied.
  The last condition is satisfied because every terminal node in $\mathcal{T}_2'$
  is also a terminal node in $\mathcal{T}_1'$,
  and we know that every terminal node in $\mathcal{T}_1'$ is consistent.
\end{proof}

\begin{lem}\label{lem:sat}
  $\varphi$ is satisfiable if and only if there exists a successful tableau for $\varphi$.
\end{lem}
\begin{proof}
  ($\implies$) Assume $\varphi$ is satisfiable, meaning that $\mathcal{M},s \models \varphi$
  for some $\mathcal{M} = (S, \rightarrow, \ell)$ and $s \in S$.
  
  Let $\mathcal{T}$ be a tableau for $\varphi$,
  and note that such a tableau always exists by applying the tableau rules to $\langle \{\varphi\}, [0,0], [0,0] \rangle$.
  Now construct a marking $\mathfrak{M} \subseteq S \times \mathcal{T}$ as follows.
  \begin{itemize}
    \item $(s,r) \in \mathfrak{M}$ where $r$ is the root of $\mathcal{T}$.
    \item If $(s',m) \in \mathfrak{M}$ and ($\land$) or ($\neg\neg$) was applied to $m$,
      add $(s',m')$ to $\mathfrak{M}$, where $m'$ is the child of $m$.
    \item If $(s',m) \in \mathfrak{M}$ and ($\neg \land$) was applied to $m$, meaning that
      \[
        \begin{prooftree}
          \hypo{m = \langle \Gamma \cup \{\neg (\varphi_1 \land \varphi_2)\}, \mathcal{I}^L, \mathcal{I}^M \rangle}
          \infer[left label = {($\neg \land$)}]1{m_1 = \langle \Gamma \cup \{\neg \varphi_1\}, \mathcal{I}^L, \mathcal{I}^M \rangle \quad m_2 = \langle \Gamma \cup \{\neg \varphi_2\}, \mathcal{I}^L, \mathcal{I}^M \rangle}
        \end{prooftree}
      \]
      then add $(s',m_1)$ to $\mathfrak{M}$ if $s' \in \sat{\neg \varphi_1}$ and add $(s',m_2)$ to $\mathfrak{M}$ if $s' \in \sat{\neg \varphi_2}$.
    \item If $(s',m) \in \mathfrak{M}$ and (mod) was applied to $m$, meaning that
      \[
        \begin{prooftree}
          \hypo{m = \langle \Gamma \cup \{N^1_{r_1} \varphi_1, \dots, N^n_{r_n} \varphi_n\} \cup \{\neg O^1_{r_1'} \varphi_1', \dots, \neg O^{n'}_{r_{n'}'} \varphi_{n'}'\}, \mathcal{I}^L, \mathcal{I}^M \rangle}
          \infer[left label = {(mod)}]1{m_1 = \langle \{\psi_1\}, \mathcal{I}^L_1, \mathcal{I}^M_1\rangle \quad \cdots \quad m_k = \langle \{\psi_k\}, \mathcal{I}^L_k, \mathcal{I}^M_k\rangle}
        \end{prooftree}
      \]
      then add $(t',m_i)$ to $\mathfrak{M}$ if $t' \in \sat{\psi_i}$ and $s' \xrightarrow{r} t'$ for some $r \in \mathbb{R}_{\geq 0}$.
  \end{itemize}
  We will first argue that for any $(s', \langle \Gamma, \mathcal{I}^L, \mathcal{I}^M \rangle) \in \mathfrak{M}$ we have $\mathcal{M},s' \models \Gamma$,
  meaning $\mathcal{M},s' \models \varphi'$ for all $\varphi' \in \Gamma$.
  We prove this by induction on the depth $d$ of $m$.
  
  $d = 0$: We have $(s', \langle \Gamma, \mathcal{I}^L, \mathcal{I}^M \rangle) = (s,r) = (s, \langle \{\varphi\}, [0,0], [0,0] \rangle)$,
    and by assumption we get $\mathcal{M},s \models \varphi$.
    
  $d > 0$: We consider which rule was applied to the parent of $m$.
    \begin{description}
      \item[($\land$)]
        \[
          \begin{prooftree}
            \hypo{m' = \langle \Gamma \cup \{\varphi_1 \land \varphi_2\}, \mathcal{I}^L, \mathcal{I}^M \rangle}
            \infer[left label = {($\land$)}]1{m = \langle \Gamma \cup \{\varphi_1, \varphi_2\}, \mathcal{I}^L, \mathcal{I}^M \rangle}
          \end{prooftree}
        \]
        By induction hypothesis, we have $\mathcal{M},s' \models \Gamma \cup \{\varphi_1 \land \varphi_2\}$,
        so $\mathcal{M},s' \models \varphi_1$ and $\mathcal{M},s' \models \varphi_2$,
        and hence $\mathcal{M},s' \models \Gamma \cup \{\varphi_1, \varphi_2\}$.
      \item[($\neg \land$)]
        \[
          \begin{prooftree}
            \hypo{m' = \langle \Gamma \cup \{\neg (\varphi_1 \land \varphi_2)\}, \mathcal{I}^L, \mathcal{I}^M \rangle}
            \infer[left label = {($\neg \land$)}]1{m_1 = \langle \Gamma \cup \{\neg \varphi_1\}, \mathcal{I}^L, \mathcal{I}^M \rangle \quad m_2 = \langle \Gamma \cup \{\neg \varphi_2\}, \mathcal{I}^L, \mathcal{I}^M \rangle}
          \end{prooftree}
        \]
        If $m = m_1$, then by the way $\mathfrak{M}$ was constructed we get $\mathcal{M},s' \models \neg \varphi_1$,
        and hence by induction hypothesis, $\mathcal{M},s' \models \Gamma \cup \{\neg \varphi_1\}$.
        Likewise we get $\mathcal{M},s' \models \Gamma \cup \{\neg \varphi_2\}$ if $m = m_2$.
      \item[($\neg\neg$)]
        \[
          \begin{prooftree}
            \hypo{m' = \langle \Gamma \cup \{\neg \neg \varphi'\}, \mathcal{I}^L, \mathcal{I}^M \rangle}
            \infer[left label = {($\neg\neg$)}]1{m = \langle \Gamma \cup \{\varphi'\}, \mathcal{I}^L, \mathcal{I}^M \rangle}
          \end{prooftree}
        \]
        By induction hypothesis we have $\mathcal{M},s' \models \Gamma \cup \{\neg \neg \varphi'\}$,
        which is equivalent to $\mathcal{M},s' \models \Gamma \cup \{\varphi'\}$.
      \item[(mod)]
        \[
          \begin{prooftree}
            \hypo{m' = \langle \Gamma \cup \{N^1_{r_1} \varphi_1, \dots, N^n_{r_n} \varphi_n\} \cup \{\neg O^1_{r_1'} \varphi_1', \dots, \neg O^{n'}_{r_{n'}'} \varphi_{n'}'\}, \mathcal{I}^L, \mathcal{I}^M \rangle}
            \infer[left label = {(mod)}]1{m_1 = \langle \{\psi_1\}, \mathcal{I}^L_1, \mathcal{I}^M_1 \rangle \quad \cdots \quad m_k = \langle \{\psi_k\}, \mathcal{I}^L_k, \mathcal{I}^M_k \rangle}
          \end{prooftree}
        \]
        We must have $m = m_i$ for some $1 \leq i \leq k$.
        By construction of $\mathfrak{M}$ we know that $\mathcal{M},m_i \models \psi_i$.
    \end{description}
    
    Now let $\mathcal{T}'$ be the subtree of $\mathcal{T}$
    consisting of those nodes $m$ where there exists a state $s'$ such that $(s',m) \in \mathfrak{M}$.
    We will now prove that $\mathcal{T}'$
    satisfies the three conditions in Definition~\ref{def:success}.
    
    For the first condition we prove the contrapositive:
    If $m$ is not a leaf in $\mathcal{T}$, then it is not a leaf in $\mathcal{T}'$.
    Hence we assume that $m$ is not a leaf in $\mathcal{T}$.
    If $m$ is not a node in $\mathcal{T}'$,
    then it is also not a leaf node in $\mathcal{T}'$.
    If $m$ is a node in $\mathcal{T}'$,
    then there must exist some state $s'$
    such that $(s',m) \in \mathfrak{M}$.
    We now consider which rule was applied to $m$ in $\mathcal{T}$.
    
    \begin{description}
      \item[($\land$) or ($\neg\neg$)] In these cases, $m$ has a child $m'$ in $\mathcal{T}$,
        and by construction of $\mathfrak{M}$, we get $(s',m') \in \mathfrak{M}$, so $m'$ is a child of $m$ in $\mathcal{T}'$.
      \item[($\neg \land$)]
        \[
          \begin{prooftree}
            \hypo{m = \langle \Gamma \cup \{\neg (\varphi_1 \land \varphi_2)\}, \mathcal{I}^L, \mathcal{I}^M\rangle}
            \infer[left label = {($\neg \land$)}]1{m_1 = \langle \Gamma \cup \{\neg \varphi_1\}, \mathcal{I}^L, \mathcal{I}^M \rangle \quad m_2 = \langle \Gamma \cup \{\neg \varphi_2\}, \mathcal{I}^L, \mathcal{I}^M \rangle}
          \end{prooftree}
        \]
        We know that $\mathcal{M},s' \models \Gamma \cup \{\neg (\varphi_1 \land \varphi_2)\}$,
        so we must have $\mathcal{M},s' \models \neg \varphi_1$ or $\mathcal{M},s' \models \neg \varphi_2$.
        By construction of $\mathfrak{M}$, this means that $(s',m_1) \in \mathfrak{M}$ or $(s',m_2) \in \mathfrak{M}$,
        and hence $m_1$ or $m_2$ must be a child of $m$ in $\mathcal{T}'$.
      \item[(mod)]
        \[
          \begin{prooftree}
            \hypo{m = \langle \Gamma \cup \{N^1_{r_1} \varphi_1, \dots, N^n_{r_n} \varphi_n\} \cup \{\neg O^1_{r_1'} \varphi_1', \dots, \neg O^{n'}_{r_{n'}'} \varphi_{n'}'\}, \mathcal{I}^L, \mathcal{I}^M \rangle}
            \infer[left label = {(mod)}]1{m_1 = \langle \{\psi_1\}, \mathcal{I}^L_1, \mathcal{I}^M_1 \rangle \quad \cdots \quad m_k = \langle \{\psi_k\}, \mathcal{I}^L_k, \mathcal{I}^M_k \rangle}
          \end{prooftree}
        \]
        For each $m_i$ there must exist some $j$ such that $N^j_{r_j} \varphi_j = N^j_{r_j} \psi_i$.
        Then we know that $\mathcal{M},s' \models N^j_{r_j} \psi_i$,
        and hence $\transl{s'}{\sat{\psi_i}} \geq r_j$ or $\transr{s'}{\sat{\psi_i}} \leq r_j$.
        In either case there must exist some $t' \in \sat{\psi_i}$
        such that $s' \xrightarrow{r} t'$ for some $r$.
        Hence $(t',m_i) \in \mathfrak{M}$ and $m_i$ is a child of $m$ in $\mathcal{T}'$.
    \end{description}
    
    For the second condition, let $(s',m) \in \mathfrak{M}$ where $m$ is a modal node, meaning that
    \[
      \begin{prooftree}
        \hypo{m = \langle \Gamma \cup \{N^1_{r_1} \varphi_1, \dots, N^n_{r_n} \varphi_n\} \cup \{\neg O^1_{r_1'} \varphi_1', \dots, \neg O^{n'}_{r_{n'}'} \varphi_{n'}'\}, \mathcal{I}^L, \mathcal{I}^M \rangle}
        \infer[left label = {(mod)}]1{m_1 = \langle \{\psi_1\}, \mathcal{I}^L_1, \mathcal{I}^M_1 \rangle \quad \cdots \quad m_k = \langle \{\psi_k\}, \mathcal{I}^L_k, \mathcal{I}^M_k \rangle}
      \end{prooftree}
    \]
    For every $\psi_i$ we must have $N^j_{r_j} \varphi_j = N^j_{r_j} \psi_i$ for some $j$,
    so $\mathcal{M},s' \models N^j_{r_j} \psi_i$,
    which implies that there exists $t' \in \sat{\psi_i}$ such that $s' \xrightarrow{r} t'$ for some $r$.
    Hence we get $(t',m_i) \in \mathfrak{M}$.
    Since this holds for any $i$, we get that every $m_i$ is included in $\mathcal{T}'$.
        
    For the third condition, let $m = \langle \Gamma, \mathcal{I}^L, \mathcal{I}^M \rangle$ be a terminal node in $\mathcal{T}'$.
    We check the conditions of Definition \ref{def:consistent}.
    There must exist a state $s'$ such that $(s',m) \in \mathfrak{M}$,
    which means that $\mathcal{M},s' \models \Gamma$.
    Hence $s'$ satisfies all the literals in $\Gamma$,
    which can only happen if the first condition is satisfied.
    For the second condition, note that $[0,0]$ is a consistent interval,
    and every interval constructed by the (mod) rule is also consistent,
    so $\mathcal{I}^L$ and $\mathcal{I}^M$ must be consistent.
    Hence it remains to check the third condition.
    Assume that
    \[\mathcal{I}^L = \lbag_1 a,b \rbag_1 \quad \text{and} \quad \mathcal{I}^M = \lbag_2 c,d \rbag_2.\]
    Now, either $\mathcal{I}^L = \mathcal{I}^M = [0,0]$,
    in which case clearly $a \leq b$, $\lbag_1 = [$, and $\rbag_2 = \; ]$,
    or there exists a modal node $m'$ in $\mathcal{T}'$ such that $m$ can be reached from $m'$.
    Let
    \[m^* = \langle \Gamma^* \cup \{N^1_{r_1}\varphi_1, \dots, N^n_{r_n}\varphi_n\} \cup \{\neg O^1_{r'_1}\varphi'_1, \dots, \neg O^{n'}_{r'_{n'}}\varphi'_{n'}\}, \mathcal{I}^L_*, \mathcal{I}^M_*\rangle\]
    be the modal node in $\mathcal{T}'$ with greatest depth from which $m$ can be reached.
    Then $m^*$ must have a child $m_i^* = \langle \{\psi_i\}, \mathcal{I}^L_i, \mathcal{I}^M_i\rangle$ where
    \[\mathcal{I}^L = \mathcal{I}^L_i \quad \text{and} \quad \mathcal{I}^M = \mathcal{I}^M_i.\]
    If $\mathcal{I}^M_i = (\max \mathbb{M}^-_i,\infty)$ or $\mathcal{I}^M_i = [0,\infty)$,
    then clearly $a < d = \infty$.
    Otherwise, if $\mathcal{I}^L_i = [0,\min\mathbb{L}^-_i)$ or $\mathcal{I}^L_i = [0,\infty)$
    and $\mathcal{I}^M_i = (\max\mathbb{M}^-_i,\min\mathbb{M}^+_i]$ or $\mathcal{I}^M_i = [0,\min\mathbb{M}^+_i]$,
    then $0 = a \leq d$, $\lbag_1 = [$, and $\rbag_2 = \; ]$.
    Otherwise, the only possibility left is that $\lbag_1 = [$, $a = \max\mathbb{L}^+_i$, $\rbag_2 = \; ]$, and $d = \min\mathbb{M}^+_i$.
    We must show that $a \leq d$.
    Assume towards a contradiction that $a = \max\mathbb{L}^+_i > \min\mathbb{M}^+_i = d$.
    Then, by the definition of $\mathbb{L}^+_i$ and $\mathbb{M}^+_i$,
    there exist $j_1$ and $j_2$ such that
    \[L_{r_1}\varphi_{j_1} = N^{j_1}_{r_{j_1}}\varphi_{j_1} \quad \text{and} \quad \varphi_{j_1} \in \upw{\Gamma'}(\psi_i)\]
    \[M_{r_2}\varphi_{j_2} = N^{j_2}_{r_{j_2}}\varphi_{j_2} \quad \text{ and } \quad \varphi_{j_2} \in \upw{\Gamma'}(\psi_i)\]
    with $r_1 > r_2$.
    Because $m^*$ is a node in $\mathcal{T}'$,
    there must exist a state $s^*$ such that $(s^*,m^*) \in \mathfrak{M}$,
    which implies that $\mathcal{M}, s^* \models L_{r_1}\varphi_{j_1}$
    and $\mathcal{M}, s^* \models M_{r_2}\varphi_{j_2}$.
    This gives us
    \[\transl{s^*}{\sat{\psi_i}} \geq \transl{s^*}{\sat{\varphi_{j_1}}} \geq r_1 > r_2 \geq \transr{s^*}{\sat{\varphi_{j_2}}} \geq \transr{s^*}{\sat{\psi_i}},\]
    which is a contradiction.
    Hence $a \leq d$ and we are done.
    
  ($\impliedby$) This follows from Lemma \ref{lem:tableaumodel}.
\end{proof}

\begin{thm}
  The satisfiability problem for our logic is decidable.
\end{thm}
\begin{proof}
  By Lemma \ref{lem:sat}, to decide whether a formula $\varphi$ is satisfiable,
  it is enough to check whether there exists a successful tableau for $\varphi$.
  Furthermore, by Lemma \ref{lem:tableaux} it is enough to only check a single tableau for $\varphi$:
  If the tableau is successful, then all tableaux for $\varphi$ are successful,
  and if it is not successful, then no tableau for $\varphi$ is successful.
  
  One can construct such a tableau for $\varphi$
  by applying the tableau rules of Table \ref{tab:rules} to the tuple $\langle \{\varphi\}, [0,0], [0,0] \rangle$
  until no more rules can be applied.
  We will now argue that there is an effective procedure for constructing such a tableau by induction on the modal depth of $\varphi$.
  
  $md(\varphi) = 0$: In this case, the (mod) rule is never used when constructing the tableau.
    Hence the procedure proceeds by syntactically checking which rules can be used at a given moment,
    and choosing a valid rule to apply.
  
  $md(\varphi) > 0$: In this case we proceed as for the case where $md(\varphi) = 0$,
    except that now the (mod) rule may also be applied,
    in which case we need to be able to compute the $\psi_i$, $\mathcal{I}^L_i$ and $\mathcal{I}^M_i$.
    The difficulty lies in computing the set $\{\psi_1, \dots, \psi_m\} = \min(\mathcal{L}(\Gamma'))$,
    where $\Gamma' = \{\varphi_1, \dots, \varphi_n\}$, and the sets
    \[\mathbb{L}^+_i = \{r \mid L_r \varphi_j = N^j_{r_j} \varphi_j \text{ for some } j \text{ and } \varphi_j \in \upw{\Gamma'}(\psi_i)\}\]
    \[\mathbb{M}^+_i = \{r \mid M_r \varphi_j = N^j_{r_j} \varphi_j \text{ for some } j \text{ and } \varphi_j \in \upw{\Gamma'}(\psi_i)\}\]
    \[\mathbb{L}^-_i = \{r \mid L_r \varphi_j' = O^j_{r_j} \varphi_j' \text{ for some } j \text{ and } \models \psi_i \rightarrow \varphi_j'\}\]
    \[\mathbb{M}^-_i = \{r \mid M_r \varphi_j' = O^j_{r_j} \varphi_j' \text{ for some } j \text{ and } \models \psi_i \rightarrow \varphi_j'\}.\]
    However, note that all $\varphi_i$ and $\varphi_i'$ and  have modal depth less than $md(\varphi)$.
    Therefore, by induction hypothesis, we have an effective procedure to decide whether
    $\models \varphi_i \rightarrow \varphi_j$ and $\models \varphi_i \leftrightarrow \varphi_j$,
    which is exactly what we need to compute the aforementioned sets.
    Given this we can compute the values needed for the intervals $\mathcal{I}^L_i$ and $\mathcal{I}^M_i$.

    The procedure for constructing a tableau for $\varphi$ uses recursion on the modal depth of $\varphi$, $md(\varphi) = k$,
    in order to compute the sets $\{\psi_1,\ldots,\psi_m\}$, $\mathbb{L}^+_i$, $\mathbb{M}^+_i$, $\mathbb{L}^-_i$, and $\mathbb{M}^-_i$.
    To compute these sets we must instantiate the procedure for constructing tableaux for formulae of modal depth $k - 1$,
    which again must instantiate the procedure for constructing tableaux for formulae of modal depth $k - 2$, and so on.
    The recursion stops when only the procedure for generating tableaux for formulae with modal depth zero is needed
    to construct the sets $\{\psi_1,\ldots,\psi_m\}$, $\mathbb{L}^+_i$, $\mathbb{M}^+_i$, $\mathbb{L}^-_i$, and $\mathbb{M}^-_i$.

    Thus, for any $k \in \mathbb{N}$, there exists a procedure for generating a tableau for any formula $\varphi$ with $md(\varphi) = k$.
    Because all formulae are finite they must have finite modal depth. Therefore, for any formula $\varphi$, there exists a procedure that generates a tableau for $\varphi$.
\end{proof}

\begin{exa}\label{ex:sat}
  Consider the formula $\varphi = \neg(\neg (L_2 p_1 \land M_5L_1 p_1) \land \neg M_2 p_2 ))$.
  Using the tableau rules, we get the following tableau $\mathcal{T}$ for $\varphi$.
  \[
    \begin{prooftree}[proof style=downwards]
      \hypo{\langle \{p_1\}, [1,\infty), [0,\infty) \rangle}
      \infer[left label = {(mod)}]1{\langle \{p_1, L_1 p_1\}, [2, \infty), [5, \infty) \rangle}
      \infer[left label = {(mod)}]1{\langle \{L_2 p_1, M_5L_1 p_1\}, [0,0], [0,0] \rangle}
      \infer[left label = {($\land$)}]1{\langle \{L_2 p_1 \land M_5L_1 p_1\}, [0,0], [0,0] \rangle}
      \infer[left label = {($\neg \neg$)}]1{\langle \{\neg\neg (L_2 p_1 \land M_5L_1 p_1)\}, [0,0], [0,0] \rangle}
      \hypo{\langle \{p_2\}, [0,\infty), [0, 2] \rangle}
      \infer[left label = {(mod)}]1{\langle \{M_2 p_2\}, [0,0], [0,0] \rangle}
      \infer[left label = {($\neg\neg$)}]1{\langle \{\neg\neg M_2 p_2\}, [0,0], [0,0] \rangle}
      \infer[left label = {($\neg \land$)}]2{\langle \{\neg (\neg (L_2 p_1 \land M_5L_1 p_1) \land M_2 p_2)\}, [0,0], [0,0] \rangle}
    \end{prooftree}
  \]
  In this case the tableau is successful,
  since all terminal nodes are consistent.
  In fact, there are three distinct subtrees witnessing this fact:
  one that chooses the left branch, one that chooses the right branch,
  and one that chooses both branches.
  In Figure \ref{fig:sat-ex} we show the resulting model $\mathcal{M}(\mathcal{T})$
  for the witness that chooses the left branch.
  
  \begin{figure}
    \begin{tikzpicture}[WTS, node distance=2cm]
      \node[state, label=above:{$\{\}$}]    (st)               {$s_\mathcal{T}$};
      \node[state, label=above:{$\{p_1\}$}] (s1) [right=of st] {$s_1$};
      \node[state, label=above:{$\{p_1\}$}] (s2) [right=of s1] {$s_2$};

      \path (s0) edge[bend right=45] node[below] {$5$} (s1);
      \path (s0) edge[bend left=45]  node[above] {$2$} (s1);

      \path (s1) edge[bend right=45] node[below] {$1$} (s2);
      \path (s1) edge[bend left=45]  node[above] {$1$} (s2);
    \end{tikzpicture}
    \captionof{figure}{The model $\mathcal{M}(\mathcal{T})$ for the successful tableau $\mathcal{T}$ in Example \ref{ex:sat}.}
    \label{fig:sat-ex}
  \end{figure}
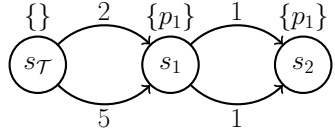
\end{exa}

\begin{exa}
  Consider the formula $\varphi = p_1 \land L_4 p_1 \land \neg L_3 p_1 \land L_2 p_2$.
  Using the tableau rules, we get the following tableau $\mathcal{T}$ for $\varphi$.
  \[
    \begin{prooftree}
      \hypo{\langle \{p_1 \land L_4 p_1 \land \neg L_3 p_1 \land L_2 p_2\}, [0,0], [0,0] \rangle}
      \infer[left label = {($\land$)}]1{\langle \{p_1, L_4 p_1 \land \neg L_3 p_1 \land L_2 p_2\}, [0,0], [0,0] \rangle}
      \infer[left label = {($\land$)}]1{\langle \{p_1, L_4 p_1, \neg L_3 p_1 \land L_2 p_2\}, [0,0], [0,0] \rangle}
      \infer[left label = {($\land$)}]1{\langle \{p_1, L_4 p_1, \neg L_3 p_1, L_2 p_2\}, [0,0], [0,0] \rangle}
      \infer[left label = {(mod)}]1{\langle \{p_1\}, [4,3), [0,\infty] \rangle \quad \langle \{p_2\}, [2, \infty), [0,\infty) \rangle}
    \end{prooftree}
  \]
  In this case the interval $[4,3)$ is not consistent,
  and hence the tableau is not successful,
  so we can conclude that $\varphi$ is not satisfiable.
\end{exa}

\section{Concluding Remarks}
Our contributions in this paper have been to define a new bisimulation relation
for weighted transition systems,
which relates those states that have similar behavior with respect to their minimum
and maximum weights on transitions,
as well as an accompanying modal logic to reason about the
upper and lower bounds of weights on transitions.
We have shown that this logic characterizes exactly
those states that are bisimilar for image-finite systems.
Furthermore, we have provided a complete axiomatization of our logic,
and we have shown that it enjoys the finite model property.
Lastly we have developed an algorithm based on the tableau method which decides the satisfiability of a formula in our logic
and constructs a finite model for the formula if it is satisfiable.

This work could be extended in different ways.
Since our logic is non-compact, strong completeness does not follow directly from weak completeness,
and hence it would be interesting to explore a strong-complete axiomatization of the proposed logic.
Such an axiomatization would need additional, infinitary axioms.
Examples of such axioms would be $\{L_q \varphi \mid q < r\} \vdash L_r \varphi$ and $\{M_q \varphi \mid q < r\} \vdash M_r \varphi$,
which are easily proven sound and describe the Archimedean property discussed in Theorem \ref{thm:noncompact}.

Although we have shown that our logic is expressive enough to capture bisimulation,
it would also be of interest to extend our logic with a kind of fixed-point operator or
standard temporal logic operators such as until in order to increase its expressivity,
and hence its practical use.
We envisage two ways in which such a logic could be given semantics:
either by accumulating weights or by taking the maximum or minimum of weights.
In the accumulating case in particular,
one could also allow negative weights to model that the system gains resources.

\subsection*{Acknowledgements.} We wish to thank the anonymous reviewers for their careful reading of our paper and for their invaluable comments that helped improve the paper. We are also grateful to Bingtian Xue for helpful discussions. This research was partially supported by the Danish FNU project 4181-00360, the ERC Advanced Grant LASSO: ``Learning, Analysis, Synthesis and Optimization of Cyber Physical Systems'' as well as the Sino-Danish Basic Research Center IDEA4CPS.

\bibliographystyle{alpha}
\bibliography{bibliography}

\newcommand{\etalchar}[1]{$^{#1}$}
\begin{thebibliography}{HLM{\etalchar{+}}16}

\bibitem[ACD93]{ALUR19932}
Rajeev Alur, Costas Courcoubetis, and David~L. Dill.
\newblock Model-checking in dense real-time.
\newblock {\em Inf. Comput.}, 104(1):2--34, 1993.

\bibitem[BDP16]{babari2016}
Parvaneh Babari, Manfred Droste, and Vitaly Perevoshchikov.
\newblock Weighted register automata and weighted logic on data words.
\newblock In Augusto Sampaio and Farn Wang, editors, {\em Theoretical Aspects
  of Computing - {ICTAC} 2016 - 13th International Colloquium, Taipei, Taiwan,
  ROC, October 24-31, 2016, Proceedings}, volume 9965 of {\em Lecture Notes in
  Computer Science}, pages 370--384, 2016.

\bibitem[BG09]{bollig2009}
Benedikt Bollig and Paul Gastin.
\newblock Weighted versus probabilistic logics.
\newblock In Volker Diekert and Dirk Nowotka, editors, {\em Developments in
  Language Theory, 13th International Conference, {DLT} 2009, Stuttgart,
  Germany, June 30 - July 3, 2009. Proceedings}, volume 5583 of {\em Lecture
  Notes in Computer Science}, pages 18--38. Springer, 2009.

\bibitem[BvBW06]{blackburn}
P.~Blackburn, J.~F. A.~K. van Benthem, and F.~Wolter.
\newblock {\em Handbook of Modal Logic}.
\newblock Studies in Logic and Practical Reasoning. Elsevier Science, 2006.

\bibitem[CK16]{katoen:sat}
Souymodip Chakraborty and Joost{-}Pieter Katoen.
\newblock On the satisfiability of some simple probabilistic logics.
\newblock In Martin Grohe, Eric Koskinen, and Natarajan Shankar, editors, {\em
  Proceedings of the 31st Annual {ACM/IEEE} Symposium on Logic in Computer
  Science, {LICS} '16, New York, NY, USA, July 5-8, 2016}, pages 56--65. {ACM},
  2016.

\bibitem[CLM11a]{cardelli2011a}
Luca Cardelli, Kim~G. Larsen, and Radu Mardare.
\newblock Continuous {M}arkovian logic - from complete axiomatization to the
  metric space of formulas.
\newblock In Marc Bezem, editor, {\em Computer Science Logic, 25th
  International Workshop / 20th Annual Conference of the EACSL, {CSL} 2011,
  September 12-15, 2011, Bergen, Norway, Proceedings}, volume~12 of {\em
  LIPIcs}, pages 144--158. Schloss Dagstuhl - Leibniz-Zentrum fuer Informatik,
  2011.

\bibitem[CLM11b]{cardelli2011b}
Luca Cardelli, Kim~G. Larsen, and Radu Mardare.
\newblock Modular {M}arkovian logic.
\newblock In Luca Aceto, Monika Henzinger, and Jir{\'{\i}} Sgall, editors, {\em
  Automata, Languages and Programming - 38th International Colloquium, {ICALP}
  2011, Zurich, Switzerland, July 4-8, 2011, Proceedings, Part {II}}, volume
  6756 of {\em Lecture Notes in Computer Science}, pages 380--391. Springer,
  2011.

\bibitem[DG05]{droste2005}
Manfred Droste and Paul Gastin.
\newblock Weighted automata and weighted logics.
\newblock In Lu{\'{\i}}s Caires, Giuseppe~F. Italiano, Lu{\'{\i}}s Monteiro,
  Catuscia Palamidessi, and Moti Yung, editors, {\em Automata, Languages and
  Programming, 32nd International Colloquium, {ICALP} 2005, Lisbon, Portugal,
  July 11-15, 2005, Proceedings}, volume 3580 of {\em Lecture Notes in Computer
  Science}, pages 513--525. Springer, 2005.

\bibitem[DR06]{droste2006a}
Manfred Droste and George Rahonis.
\newblock Weighted automata and weighted logics on infinite words.
\newblock In Oscar~H. Ibarra and Zhe Dang, editors, {\em Developments in
  Language Theory, 10th International Conference, {DLT} 2006, Santa Barbara,
  CA, USA, June 26-29, 2006, Proceedings}, volume 4036 of {\em Lecture Notes in
  Computer Science}, pages 49--58. Springer, 2006.

\bibitem[DV06]{droste2006b}
Manfred Droste and Heiko Vogler.
\newblock Weighted tree automata and weighted logics.
\newblock {\em Theor. Comput. Sci.}, 366(3):228--247, 2006.

\bibitem[{\'{E}}si14]{esik2014}
Zolt{\'{a}}n {\'{E}}sik.
\newblock Axiomatizing weighted synchronization trees and weighted
  bisimilarity.
\newblock {\em Theor. Comput. Sci.}, 534:2--23, 2014.

\bibitem[FH94]{Fagin}
Ronald Fagin and Joseph~Y. Halpern.
\newblock Reasoning about knowledge and probability.
\newblock {\em J. {ACM}}, 41(2):340--367, 1994.

\bibitem[Fic11]{fichtner2011}
Ina Fichtner.
\newblock Weighted picture automata and weighted logics.
\newblock {\em Theory Comput. Syst.}, 48(1):48--78, 2011.

\bibitem[GH09]{halmos2009}
Steven Givant and Paul Halmos.
\newblock {\em Introduction to Boolean Algebras}.
\newblock Undergraduate Texts in Mathematics. Springer, 2009.

\bibitem[HLM{\etalchar{+}}16]{hansen2016}
Mikkel Hansen, Kim~Guldstrand Larsen, Radu Mardare, Mathias~Ruggaard Pedersen,
  and Bingtian Xue.
\newblock A complete approximation theory for weighted transition systems.
\newblock In Martin Fr{\"{a}}nzle, Deepak Kapur, and Naijun Zhan, editors, {\em
  Dependable Software Engineering: Theories, Tools, and Applications - Second
  International Symposium, {SETTA} 2016, Beijing, China, November 9-11, 2016,
  Proceedings}, volume 9984 of {\em Lecture Notes in Computer Science}, pages
  213--228, 2016.

\bibitem[HM01]{Heifetz200131}
Aviad Heifetz and Philippe Mongin.
\newblock Probability logic for type spaces.
\newblock {\em Games and Economic Behavior}, 35(1-2):31--53, 2001.

\bibitem[JL91]{JonssonL91}
Bengt Jonsson and Kim~Guldstrand Larsen.
\newblock Specification and refinement of probabilistic processes.
\newblock In {\em Proceedings of the Sixth Annual Symposium on Logic in
  Computer Science {(LICS} '91), Amsterdam, The Netherlands, July 15-18, 1991},
  pages 266--277. {IEEE} Computer Society, 1991.

\bibitem[JLMX14]{DBLP:journals/entcs/JaziriLMX14}
Samy Jaziri, Kim~Guldstrand Larsen, Radu Mardare, and Bingtian Xue.
\newblock Adequacy and complete axiomatization for timed modal logic.
\newblock {\em Electr. Notes Theor. Comput. Sci.}, 308:183--210, 2014.

\bibitem[JLS12]{Juhl2012408}
Line Juhl, Kim~G. Larsen, and Jir{\'{\i}} Srba.
\newblock Modal transition systems with weight intervals.
\newblock {\em J. Log. Algebr. Program.}, 81(4):408--421, 2012.

\bibitem[KLMP13]{6571564}
Dexter Kozen, Kim~G. Larsen, Radu Mardare, and Prakash Panangaden.
\newblock Stone duality for {M}arkov processes.
\newblock In {\em 28th Annual {ACM/IEEE} Symposium on Logic in Computer
  Science, {LICS} 2013, New Orleans, LA, USA, June 25-28, 2013}, pages
  321--330. {IEEE} Computer Society, 2013.

\bibitem[KMP13]{KozenMP13}
Dexter Kozen, Radu Mardare, and Prakash Panangaden.
\newblock Strong completeness for {M}arkovian logics.
\newblock In Krishnendu Chatterjee and Jir{\'{\i}} Sgall, editors, {\em
  Mathematical Foundations of Computer Science 2013 - 38th International
  Symposium, {MFCS} 2013, Klosterneuburg, Austria, August 26-30, 2013.
  Proceedings}, volume 8087 of {\em Lecture Notes in Computer Science}, pages
  655--666. Springer, 2013.

\bibitem[LM14]{larsen1}
Kim~G. Larsen and Radu Mardare.
\newblock Complete proof systems for weighted modal logic.
\newblock {\em Theor. Comput. Sci.}, 546:164--175, 2014.

\bibitem[LMX14a]{larsen2014b}
Kim~Guldstrand Larsen, Radu Mardare, and Bingtian Xue.
\newblock Decidability and expressiveness of recursive weighted logic.
\newblock In Andrei Voronkov and Irina Virbitskaite, editors, {\em Perspectives
  of System Informatics - 9th International Ershov Informatics Conference,
  {PSI} 2014, St. Petersburg, Russia, June 24-27, 2014. Revised Selected
  Papers}, volume 8974 of {\em Lecture Notes in Computer Science}, pages
  216--231. Springer, 2014.

\bibitem[LMX14b]{larsen2014a}
Kim~Guldstrand Larsen, Radu Mardare, and Bingtian Xue.
\newblock A decidable recursive logic for weighted transition systems.
\newblock In Gabriel Ciobanu and Dominique M{\'{e}}ry, editors, {\em
  Theoretical Aspects of Computing - {ICTAC} 2014 - 11th International
  Colloquium, Bucharest, Romania, September 17-19, 2014. Proceedings}, volume
  8687 of {\em Lecture Notes in Computer Science}, pages 460--476. Springer,
  2014.

\bibitem[LMX15a]{larsen2015}
Kim~G. Larsen, Radu Mardare, and Bingtian Xue.
\newblock Alternation-free weighted mu-calculus: Decidability and completeness.
\newblock {\em Electr. Notes Theor. Comput. Sci.}, 319:289--313, 2015.

\bibitem[LMX15b]{LarsenMX15}
Kim~G. Larsen, Radu Mardare, and Bingtian Xue.
\newblock Concurrent weighted logic.
\newblock {\em J. Log. Algebr. Meth. Program.}, 84(6):884--897, 2015.

\bibitem[LMX18]{Larsen2016}
Kim~G. Larsen, Radu Mardare, and Bingtian Xue.
\newblock On decidability of recursive weighted logics.
\newblock {\em Soft Comput.}, 22(4):1085--1102, 2018.

\bibitem[LS91]{probabilistic_bisimulation}
Kim~Guldstrand Larsen and Arne Skou.
\newblock Bisimulation through probabilistic testing.
\newblock {\em Inf. Comput.}, 94(1):1--28, 1991.

\bibitem[MCL12]{mardare2012}
Radu Mardare, Luca Cardelli, and Kim~G. Larsen.
\newblock Continuous {M}arkovian logics - axiomatization and quantified
  metatheory.
\newblock {\em Logical Methods in Computer Science}, 8(4), 2012.

\bibitem[Mei06]{meinecke2006}
Ingmar Meinecke.
\newblock Weighted logics for traces.
\newblock In Dima Grigoriev, John Harrison, and Edward~A. Hirsch, editors, {\em
  Computer Science - Theory and Applications, First International Computer
  Science Symposium in Russia, {CSR} 2006, St. Petersburg, Russia, June 8-12,
  2006, Proceedings}, volume 3967 of {\em Lecture Notes in Computer Science},
  pages 235--246. Springer, 2006.

\bibitem[Zho09]{Zhou09}
Chunlai Zhou.
\newblock A complete deductive system for probability logic.
\newblock {\em J. Log. Comput.}, 19(6):1427--1454, 2009.

\end{thebibliography}
\end{document}